\documentclass[runningheads]{llncs}
\usepackage[T1]{fontenc}
\usepackage{graphicx}

\usepackage{amssymb}
\usepackage{amsmath}
\usepackage{hyperref}
\usepackage{xcolor}
\usepackage{enumerate}
\usepackage{listings}
\usepackage{microtype}
\usepackage{stackengine}
\usepackage{array}
\usepackage[table,xcdraw]{xcolor}
\usepackage{enumitem}
\usepackage{subcaption}
\usepackage{cleveref}
\usepackage{thmtools}
\usepackage{QED}
\usepackage{orcidlink}

\usepackage{apptools}
\newcommand{\restateref}[1]{\IfAppendix{\hyperref[#1]{$\star$}}{\hyperref[#1*]{$\star$}}}

\setlist[description]{leftmargin=.8em, labelindent=0em}

\crefname{theorem}{Theorem}{Theorems}
\crefname{lemma}{Lemma}{Lemmas}
\crefname{observation}{Observation}{Observations}
\crefname{corollary}{Corollary}{Corollaries}
\crefname{proposition}{Proposition}{Propositions}
\crefname{claim}{Claim}{Claims}
\crefname{section}{Section}{Sections}
\crefname{appendix}{Appendix}{Appendices}
\crefname{figure}{Figure}{Figures}
\crefname{table}{Table}{Tables}
\DeclareMathOperator{\CH}{CH}

\begin{document}
\title{Edge-Constrained Hamiltonian Paths\\ on a Point Set}

\author{Todor {Anti\'c}\inst{1}\orcidlink{0009-0008-6521-7987}\and
Aleksa {D\v{z}uklevski}\inst{1}\orcidlink{0009-0003-3028-4243}\and
Ji\v{r}\'i {Fiala}\inst{1}\orcidlink{0000-0002-8108-567X} \and
Jan {Kratochv\'il}\inst{1}\orcidlink{0000-0002-2620-6133}\and
Giuseppe Liotta\inst{2}\orcidlink{0000-0002-2886-9694}  \and
Morteza Saghafian\inst{3}\orcidlink{0000-0002-4201-5775} \and
Maria Saumell\inst{4}\orcidlink{0000-0002-4704-2609} \and
Johannes Zink\inst{5}\orcidlink{0000-0002-7398-718X}}
\authorrunning{Anti\'c, D\v{z}uklevski, Fiala, Kratochvíl, Liotta, Saghafian, Saumell, and Zink}
\institute{%
Department of Applied Mathematics, Faculty of Mathematics and Physics, Charles University, Czech Republic
\email{\{todor,fiala,honza\}@kam.mff.cuni.cz, aleksa@matfyz.cz} \and
Universit\`a degli Studi di Perugia, Perugia, Italy
\email{giuseppe.liotta@unipg.it} \and
Institute of Science and Technology Austria, Klosterneuburg, Austria
\email{morteza.saghafian@ist.ac.at} \and
Department of Theoretical Computer Science, Faculty of Information Technology, Czech Technical University in Prague, Czech Republic
\email{maria.saumell@fit.cvut.cz} \and
Technische Universit\"at M\"unchen, Heilbronn, Germany\\
\email{zink@algo.cit.tum.de}}
\maketitle
\begin{abstract}
Let $S$ be a set of distinct points in general position in the Euclidean plane. A plane Hamiltonian path on $S$ is a crossing-free geometric path such that every point of $S$ is a vertex of the path. It is known that, if $S$ is sufficiently large, there exist three edge-disjoint plane Hamiltonian paths on $S$. In this paper we study an edge-constrained version of the problem of finding Hamiltonian paths on a point set. We first consider the problem of finding a single plane Hamiltonian path $\pi$ with endpoints $s,t \in S$  and constraints given by a segment $\overline{ab}$, where $a,b\in S$. We consider the following scenarios: (i)~$\overline{ab} \in \pi$; (ii)~$\overline{ab}\not\in \pi$. We characterize those quintuples $\langle S,a,b,s,t\rangle$ for which $\pi$ exists.  
Secondly, we consider the problem of finding two plane Hamiltonian paths $\pi_1,\pi_2$ on a set $S$ with constraints given by a segment $\overline{ab}$, where $a,b\in S$. We consider the following scenarios: (i)~$\pi_1$ and $\pi_2$ share no edges and $\overline{ab}$ is an edge of $\pi_1$; (ii)~$\pi_1$ and $\pi_2$ share no edges and none of them includes $\overline{ab}$ as an edge; (iii)~both $\pi_1$ and $\pi_2$ include $\overline{ab}$ as an edge and share no other edges. In all cases, we characterize those triples $\langle S,a,b \rangle$ for which $\pi_1$ and $\pi_2$ exist.

\end{abstract}
\section{Introduction}

Let $S$ be a set of points in the plane in \empty{general position}, that is, with no three collinear points.
A \emph{geometric graph} on $S$ is a graph with vertex set $S$ whose edges are segments between points in $S$.
A geometric graph on $S$ is \emph{plane} if no two edges cross each other.
A \emph{Hamiltonian path} on $S$ is a geometric graph on $S$ that is a path and contains all points from~$S$.
Our goal is to compute one or two \emph{plane Hamiltonian paths} on $S$ that satisfy certain \emph{geometric constraints}.
In the case where the constraint is that the two Hamiltonian paths are \emph{edge-disjoint},
the existence of two such paths was shown by Aichholzer, Hackl, Korman, van Kreveld, L\" offler, Pilz, Speckmann and Welzl~\cite{AichholzerHKKLPSW17}.
This result was recently extended by Kindermann, Kratochv\' il, Liotta and Valtr~\cite{KindermannKLV23},
who showed that in fact three edge-disjoint plane Hamiltonian paths on~$S$ always exist if $|S|$ is large enough. 

When working on problems involving point sets in the plane, we might need some different or additional properties that the paths need to satisfy besides being plane, Hamiltonian, and edge-disjoint.
We might want to find a path with prescribed endpoints or a prescribed (first) edge. In some cases, finding paths with such properties is straightforward and serves more as an exercise, while in others it can pose a significant challenge and requires more complex proofs. Such problems often arise when studying reconfigurations of plane paths, see \cite[Lemmas 1 and 2]{Aicholzer2022} or \cite[Lemma 9]{Kleist2024}. In a different direction, determining if it is possible to find a single plane Hamiltonian path that does not contain a prescribed set of segments has been studied by several authors \cite{Keller2017,Keller2019,ern2004,ern2007}, but the exact solution is known only when the point set is in convex position.   

In this paper, we investigate the problem with the natural constraints that segments or endpoints are given and must be included or avoided.
We restrict to the case where just one segment is given,
which already turns out to be non-trivial.
We remark that the existence of edge-restricted plane Hamiltonian paths has already been studied in the more general setting of \emph{simple drawings} of complete graphs (see, for example, \cite[Theorem~3.13]{Aichholzer2024} or \cite[Theorem~2.4]{Helena}), but these results are not useful for the problems that we study.

Our problems are also related to a well-studied topic of geometric graph theory, called \emph{geometric graph packing problem}. Given a geometric \emph{host} graph $G$ and a subgraph $H$ of~$G$, we want to find as many edge-disjoint copies of $H$ inside $G$ as possible. In our case, the host graph is the complete geometric graph with vertex set $S$ and $H$ is a plane Hamiltonian path, that is, a plane spanning path of $G$.
Bose, Hurtado, Rivera-Campo, and Wood~\cite{Bose200471,DBLP:journals/comgeo/BoseHRW06} characterize those plane trees that can be packed into a complete geometric graph whose vertex set is in convex position. The authors of~\cite{AichholzerHKKLPSW17,Aichholzer2014233} prove that $\Omega(\sqrt{n})$ plane trees can be packed into a complete geometric graph with~$n$ vertices. This was later improved to $\lfloor n/3\rfloor$ by Biniaz and Garc\'{i}a \cite{DBLP:conf/cccg/BiniazG18,Biniaz2020}. However, it remains 
an open question 
whether this lower bound also applies to plane Hamiltonian paths. Biniaz, Bose, Maheshwari, and Smid~\cite{DBLP:journals/dmtcs/BiniazBMS15} show that any set of~$n$ points in general position admits at least $\lceil \log_2 n \rceil - 1$ plane perfect matchings, and that, for some point sets, the maximum number of such matchings is at most $\lceil n/3 \rceil$.

In the non-geometric setting,
our problem is related to a line of research investigating
a generalization of Hamiltonian paths on planar graphs that are called \emph{Tutte paths}~\cite{DBLP:conf/icalp/BiedlK19,DBLP:journals/jgt/Sanders97,DBLP:conf/icalp/0003S18,DBLP:journals/jgt/Thomassen83a,Tutte1977}.
Tutte paths are paths in planar graphs that decompose the graph into components with at most three attachments to the path.
In some scenarios, it is asked for a Tutte path that starts and ends at given vertices and contains a given edge.

\paragraph{Our contribution.}

We study the existence of one or two plane Hamiltonian paths under various geometric constraints.
We first consider the problem of finding a single plane Hamiltonian path $\pi$ with endpoints $s,t \in S$  and constraints given by a segment $\overline{ab}$, where $a,b\in S$; see \cref{sec:onepath}.
We examine the following scenarios: (i)~$\overline{ab}\not\in \pi$ (\cref{thm:stpathsegnotinc}) and (ii)~$\overline{ab} \in \pi$ (\cref{thm:stpathabincluded}).
We characterize those quintuples $\langle S,a,b,s,t\rangle$ for which $\pi$ exists.

Secondly, as our main results,
we consider the problem of finding two plane Hamiltonian paths $\pi_1,\pi_2$ with constraints given by a segment $\overline{ab}$, where $a,b\in S$;
see \cref{sec:twopaths}.
We study the following scenarios:
(i)~$\pi_1$ and $\pi_2$ share no edges and none of them includes $\overline{ab}$ as an edge (\cref{prop:twopathssegnotinc}),
(ii)~$\pi_1$ and $\pi_2$ share no edges and $\overline{ab}$ is an edge of $\pi_1$ (\cref{Thm:OneEdgeOnePath}), and
(iii)~both~$\pi_1$ and $\pi_2$ include $\overline{ab}$ as an edge and share no other edges (\cref{thm:OneEdgeBothPaths}).
In all cases, we characterize those triples $\langle S,a,b \rangle$ for which $\pi_1$ and $\pi_2$ exist.

Statements whose proofs are fully or partially moved to the Appendix are marked with a clickable~($\star$).

\section{Preliminaries} \label{sec:prelims}

Throughout the paper we always assume that $S$ is a set of $n$ points in the plane in general position.
By general position, we mean that there are no three collinear points.
We write $\CH(S)$ for the convex hull of $S$ and $\partial \CH(S)$ for its boundary. 
For two points $a,b$, we write $\ell(ab)$ to denote the line through $a$ and~$b$, and $\overline{ab}$ to denote the line segment between $a$ and~$b$.
Given four distinct points $a,b,x,y$, we call a line segment $\overline{xy}$  a \emph{bridge} over $\ell(ab)$ if it crosses $\ell(ab)$ but not~$\overline{ab}$. If $\pi_1$ and $\pi_2$ are two (Hamiltonian) paths on $S$, we write $\pi_1\cap \pi_2=\emptyset$ if they are edge-disjoint.
This is somewhat unusual, 
but convenient for our application.
Usually, we call $v\in S$ a point, but we call a point $v\in S \cap \partial\CH(S)$ a \emph{vertex} of~$\partial\CH(S)$ and a segment between two vertices of~$\partial\CH(S)$ an \emph{edge} of~$\CH(S)$.
We often refer to two results from~\cite{KindermannKLV23} stated next. 
Inspecting their constructive proofs from an algorithmic point of view, it is easy to see that the described procedures can be realized
as $O(n \log n)$-time algorithms.
The running time comes from sorting a linear number of points by their radial coordinates.

\begin{lemma}[\cite{KindermannKLV23}, Lemma 5]
    \label{lem:pathTwoEndpoints}
    Let $S$ be a set of points in general position in the plane, and
    let $s$ and $t$ be two distinct points in~$S$.
    Then, there is a plane Hamiltonian path $\pi$ in $S$ such that $\pi$
    starts in~$s$ and ends in~$t$.
    Such a path can be found in $O(n \log n)$ time.
\end{lemma}

\begin{lemma}[\cite{KindermannKLV23}, Theorem 2] \label{Lem: KKGVThm2}
    Let $S$ be a set of at least 5 points in general position in the plane
    and let $s$ and $t$ be two (not necessarily distinct) points of $\partial \CH(S)$. Then, $S$ contains two edge-disjoint plane Hamiltonian paths, one
    starting at $s$ and the other one at $t$. Moreover, if the points $s$ and $t$ are distinct, then the paths can be chosen so that none of them contains the edge $st$.
    Such paths can be found in $O(n \log n)$ time.
\end{lemma}

For our purposes, \cref{Lem: KKGVThm2} is often too strong and cannot be applied to four or fewer points.
We prove a weaker statement that holds already for four points:

\begin{lemma} \label{Lem: S=4}
    Let $S$ be a set of at least 4 points in general position in the plane 
    and let $s$ and~$t$ be two consecutive points of $\partial \CH(S)$. Then, $S$ contains two edge-disjoint plane Hamiltonian paths, one starting at $s$ and the other one at $t$.
    Such paths can be found in $O(n \log n)$ time.   
\end{lemma}

\begin{proof}
    If $|S|\ge 5$, the result follows from \cref{Lem: KKGVThm2}.
    If $|S|=4$, there are only two sets to consider (see \cref{fig:twopaths4vertices}), and for both we can find the desired paths.
\end{proof}

\begin{figure}[t]
    \centering
    \includegraphics{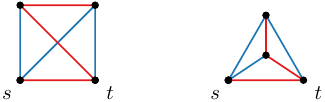}
    \caption{Two
        Hamiltonian paths on $4$ points with neighboring starting points $s,t$.}\label{fig:twopaths4vertices}
\end{figure}

Often, we use the following lemma to construct a path:

\begin{restatable}[\restateref{lem:pathGivenSegment}]{lemma}{pathGivenSegment}
    \label{lem:pathGivenSegment}
    Let $S$ be a set of at least three points in general position in the plane, and let $a, b, z \in S$ be distinct with $z \in \partial \CH(S)$.  
    Then, there is a plane Hamiltonian path $\pi$ in $S$ such that $\overline{ab} \in \pi$ and $\pi$ starts in~$z$.  
    Such a path can be found in $O(n \log n)$ time.
\end{restatable}

\section{Finding one plane Hamiltonian path}\label{sec:onepath}

In this section, we consider the following problem: We are given a point set~$S$ and two distinct points $s,t\in S$. Our goal is to find a plane Hamiltonian $s$--$t$ path $\pi$ in $S$. Additionally, we are given two distinct points $a,b\in S$ and we require that $\pi$ contains or avoids the segment $\overline{ab}$. In each of these two situations, we fully characterize the quintuples $(S,a,b,s,t)$ for which such a path exists.

\subsection{One path with prescribed endpoints and avoiding a segment}\label{sec:proofofstpathsegnotinc}

We consider the problem of finding an $s$--$t$ path $\pi$ that does not contain $\overline{ab}$. To this end, we first observe that, in some situations, such a path does not exist. The proofs of the following two lemmas describing such situations can be found in \cref{app:proofofstpathsegnotinc}.
See \cref{fig:obstaclesstpathsegnotincluded} for examples of sets satisfying the conditions of  \cref{lem:convexobstacle,lem:wheelobstacle}.

\begin{restatable}[\restateref{lem:convexobstacle}]{lemma}{convexobstacle}
\label{lem:convexobstacle}
Let $S$ be a set of $n \ge 3$ points in convex position and $a,b,s,t \in S$ points such that $\overline{ab}$ and $\overline{st}$ are disjoint edges of $\partial \CH(S)$. Then, $\pi$ does not exist.
\end{restatable}

\begin{restatable}[\restateref{lem:wheelobstacle}]{lemma}{wheelobstacle}\label{lem:wheelobstacle}
    Let $S$ be a set of $n\ge 5$ points with $n-1$ vertices on $\partial\CH(S)$. Let $s$ be the point in the interior of $\CH(S)$, $a,b,t$ be three distinct vertices of $\partial\CH(S)$ and let $x,y$ be the two vertices of $\partial\CH(S)$ adjacent to $t$ such that $a,b,x,t,y$ appear in this order on $\partial\CH(S)$. If $\overline{ab}$ is an edge of $\partial\CH(S)$ and $s$ lies on the same side of $\ell(ax)$ and $\ell(by)$ as $t$, then $\pi$ does not exist.   
\end{restatable}

\begin{figure}[t]
    \centering
    \includegraphics{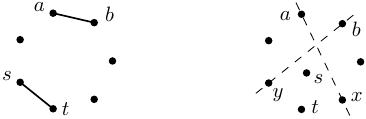}
    \caption{Sets satisfying the conditions of \cref{lem:convexobstacle} (left) and \cref{lem:wheelobstacle} (right).}
    \label{fig:obstaclesstpathsegnotincluded}
\end{figure}

The main result of this section is that, except in the cases  described above, the desired path $\pi$ always exists.

\begin{restatable}[\restateref{thm:stpathsegnotinc}]{theorem}{stpathsegnotinc}
    \label{thm:stpathsegnotinc}
    Let $S$ be a set of at least 5 points in general position in the plane and $a,b,s,t\in S$ such that $a \ne b$, $s \ne t$. Then, there is a plane Hamiltonian path~$\pi$ in $S$ with endpoints $s$ and $t$ such that $\overline{ab} \not\in \pi$, unless the quintuple $\langle S,a,b,s,t\rangle$ satisfies the conditions of \cref{lem:convexobstacle} or \cref{lem:wheelobstacle}.
    We can decide if $\pi$ exists, and construct it if it does, in $O(n\log n)$ time, where $n=|S|$.
\end{restatable}

To prove \cref{thm:stpathsegnotinc}, we will need the following less restricted statement.

\begin{restatable}[\restateref{lem:oneendptsegnotinc}]{lemma}{oneendptsegnotinc}
    \label{lem:oneendptsegnotinc}
    Let $S$ be a set of $n \ge 4$ points in general position in the plane, $a \ne b \in S$, and let $s \in S$ be any point.
    Then, there is a plane Hamiltonian path $\pi$ in $S$ such that one endpoint of $\pi$ is $s$, the other endpoint is $a$ or $b$, and $\overline{ab} \not\in \pi$.
    Such a path can be computed in $O(n \log n)$ time.
\end{restatable}

\begin{proof}[of \cref{thm:stpathsegnotinc}]
Assume that $n\ge 5$ and the quintuple $\langle S,a,b,s,t\rangle$ does not satisfy the conditions of \cref{lem:convexobstacle} or \cref{lem:wheelobstacle}.
Further, assume that $\overline{ab}$ is horizontal and that $a$ lies to the left of $b$. Also, assume that $\{a,b\} \ne \{s,t\}$, as otherwise the result follows from \cref{lem:pathTwoEndpoints}. We now split into cases.

\begin{description}
    \item[Case A: $s$ and $t$ lie in different closed halfplanes determined by $\ell(ab)$.]
    Note that the assumption of this case guarantees that $s,t\not\in \{a,b\}$.  Without loss of generality we  assume that $s$ lies above and $t$ lies below $\ell(ab)$, and that $s$ is not the only point of $S$ lying above $\ell(ab)$. 
    We first apply \cref{lem:oneendptsegnotinc} and construct a plane path starting in $s$ and collecting all the points of $S$ that are in the same closed halfplane determined by $\ell(ab)$ as $s$, but not using the segment~$\overline{ab}$. Such path is either an $s$--$a$ or an $s$--$b$ path, so assume that the former is true. Then, we construct an $a$--$t$ path for the remaining points using \cref{lem:pathTwoEndpoints}. Finally, concatenating these two paths gives $\pi$.  See \cref{fig:onepathsegnotinc-1}.
    
    \item[Case B: $s$ and $t$ lie in the same closed halfplane determined by $\ell(ab)$] \phantom{linebreak}
    \textbf{and $\overline{ab}$ is not an edge of $\partial\CH(S).$}~
    See \cref{app:proofofstpathsegnotinc}.
    
    \item[] For the remaining cases, we assume that $\overline{ab}$ is an edge of $\partial\CH(S)$.
    
    \item[Case C:  $S$ is in convex position and $\overline{st}$ is not an edge of $\partial\CH(S)$.] \phantom{linebreak}
    See \cref{app:proofofstpathsegnotinc} and \cref{fig:onepathsegnotinc-2}.

    \item[Case D:  $S$ contains only one point in the interior of $\CH(S)$.] \phantom{linebreaklinebreak}
    See \cref{app:proofofstpathsegnotinc} and \cref{fig:onepathsegnotinc-3}.

    \item[Case E: $S$ contains at least two points in the interior of $\CH(S)$.]\phantom{linebreak}
    If the only two points in the interior of $\CH(S)$ are $s$ and $t$, $\pi$ can be obtained from 
    the $a$--$b$ path in $\partial\CH(S)$ not containing $\overline{ab}$, by appending one of $s, t$ to the front and the other to the end. At least one of the two possible cases yields a non-crossing path (if the segments $\overline{sa}$ and $\overline{bt}$ cross, $\overline{ta}$ and $\overline{bs}$ are disjoint).
    
    From now on, we may assume that there is at least one point distinct from $s,t$ in the interior. Let $x$ be a point in the interior of $\CH(S)$ such that the triangle $\triangle(a,b,x)$ contains no other points of $S$. We now need to consider two subcases. 
    
    \item[Case E1: We can choose $x\neq s$.] Then, the lines $\ell(ax)$ and $\ell(bx)$ split $\CH(S)$ into $4$ convex sets. If one of them contains both $s$ and $t$, we split it into two using a suitable ray from $x$ in such a way that $s$ and $t$ are in separate convex sets. We then apply \cref{lem:pathTwoEndpoints} (at most) three times to obtain suitable $s$--$a$ and $t$--$b$ paths\footnote{Or suitable $t$--$a$ and $s$--$b$ paths.} and concatenating them with the segments $\overline{ax}$ and $\overline{xb}$ finally gives us the desired path $\pi$.
    Observe that, in the construction of the $t$--$b$ path (or, depending on the situation, one of the other paths),
    we can always find a suitable bridge over~$\ell(bx)$ (or $\ell(ax)$);
    see \cref{fig:onepathsegnotinc-4}.

    \item[Case E2: $s$ is the only point such that the triangle $\triangle(a,s,b)$ is empty.] See \cref{app:proofofstpathsegnotinc} and \cref{fig:onepathsegnotinc-5,fig:onepathsegnotinc-6}.
        
\end{description}
    
Note that we only apply \cref{lem:pathTwoEndpoints} and make local modifications or simple tests for the cases,
which allows for a running time in $O(n \log n)$.
\end{proof}

\begin{figure}[t]
    \centering
    \begin{subfigure}[t]{0.3\textwidth}
        \centering
        \includegraphics[page=1]{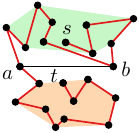}
        \subcaption{Case A.}
        \label{fig:onepathsegnotinc-1}
    \end{subfigure}
    \hfill
    \begin{subfigure}[t]{0.3\textwidth}
        \centering
        \includegraphics[page=2]{Figures/onepathsegnotinc.2.pdf}
        \subcaption{Case C.}
        \label{fig:onepathsegnotinc-2}
    \end{subfigure}
    \hfill
    \begin{subfigure}[t]{0.3\textwidth}
        \centering
        \includegraphics[page=3]{Figures/onepathsegnotinc.2.pdf}
        \subcaption{Case D.}
        \label{fig:onepathsegnotinc-3}
    \end{subfigure}
    
    \vspace{1em}
    
    \begin{subfigure}[t]{0.3\textwidth}
        \centering
        \includegraphics[page=4]{Figures/onepathsegnotinc.2.pdf}
        \subcaption{Case E1 -- we apply \cref{lem:pathTwoEndpoints} in the three colored regions.}
        \label{fig:onepathsegnotinc-4}
    \end{subfigure}
    \hfill
    \begin{subfigure}[t]{0.3\textwidth}
        \centering
        \includegraphics[page=6]{Figures/onepathsegnotinc.2.pdf}
        \subcaption{Case E2 -- construction when $t$ is an interior point of $S$.}
        \label{fig:onepathsegnotinc-5}
    \end{subfigure}
    \hfill
    \begin{subfigure}[t]{0.3\textwidth}
        \centering
        \includegraphics[page=5]{Figures/onepathsegnotinc.2.pdf}
        \subcaption{Case E2 -- construction when $t$ is a vertex of $\partial\CH(S)$.}
        \label{fig:onepathsegnotinc-6}
    \end{subfigure}
    
    \caption{Construction of path $\pi$ in \cref{thm:stpathsegnotinc}.}
    \label{fig:onepathsegnotinc}
\end{figure}

\subsection{One path with prescribed endpoints and including a segment}
\label{sec:proofofstpathabincluded}

Next, we consider the problem of finding an $s$--$t$ path $\pi$ that contains $\overline{ab}$. As in the previous problem, in some situations such a path does not exist. The first claim is immediate.

\begin{lemma}
    \label{lem:stisabobstruction}
    If $S$ is a set of more than two points in general position and $\{s,t\}=\{a,b\}$, then $\pi$ does not exist.
\end{lemma}

The proofs of the next three lemmas can be found in \cref{app:proofofstpathabincluded}.

\begin{restatable}[\restateref{lem:diagonalobstruction}]{lemma}{diagonalobstruction}
   \label{lem:diagonalobstruction}
    If $S$ is a set of $n$ points in general position, $a,b$ are non-consecutive vertices of $\partial\CH(S)$, and $s$ and $t$ lie in the same closed halfplane determined by $\ell(ab)$, then $\pi$ does not exist.
\end{restatable}

Before describing the remaining situations in which $\pi$ does not exist, recall that a bridge over $\ell(ab)$ is a segment that crosses $\ell(ab)$ but not $\overline{ab}$.

\begin{restatable}[\restateref{lem:bottomobstruction}]{lemma}{bottomobstruction}
   \label{lem:bottomobstruction}
    Let $S$ be a set of $n$ points in general position. Suppose that $\{a,b\} \not\subset \partial\CH(S)$, one of $s,t$ is in $\{a,b\}$, in the open halfplane determined by $\ell(ab)$ containing the other of $s,t$ there is at least one more point of $S$, and every bridge $\overline{xy}$ over $\ell(ab)$ is incident to the other of $s,t$. Then, $\pi$ does not exist.
\end{restatable}

\begin{restatable}[\restateref{lem:complicatedobstruction}]{lemma}{complicatedobstruction}
   \label{lem:complicatedobstruction}
    Let $S$ be a set of $n$ points in general position. Suppose that $\overline{ab}$ is not an edge of $\partial\CH(S)$, $s\in\{a,b\}$,
    and for every bridge $\overline{xy}$ over $\ell(ab)$, the following holds:
    
    \begin{itemize}
        \item both $x$ and $y$ are vertices of $\partial\CH(S)$,
        \item if $\overline{xy}$ is an edge of $\partial\CH(S)$, then $t \in \{x,y\}$ and
        \item the open halfplane determined by $\ell(xy)$ that contains $a,b$ also contains at least one point on each side of $\ell(ab)$. 
    \end{itemize}
    
    Then, $\pi$ does not exist.
    
\end{restatable}

See \cref{fig:complicatedobstruction} for examples of sets satisfying the conditions of \cref{lem:bottomobstruction,lem:complicatedobstruction}. 

\begin{figure}[t]
    \centering
    \includegraphics[page=2]{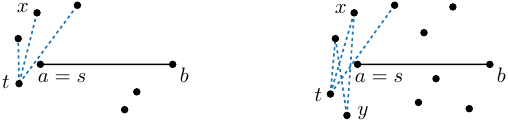}
    \caption{Sets satisfying the conditions of \cref{lem:bottomobstruction} (left) and \cref{lem:complicatedobstruction} (right), with possible bridges drawn in dashed blue.}
    \label{fig:complicatedobstruction}
\end{figure}

The main result of this section is that, except in the cases described above, $\pi$ always exist.
We give the details in \cref{app:proofofstpathabincluded}.

\begin{restatable}[\restateref{thm:stpathabincluded}]{theorem}{stpathabincluded}
    \label{thm:stpathabincluded}
    Let $S$ be a set of at least 2 points in general position in the plane and $a,b,s,t\in S$ such that $a \ne b$, $s \ne t$. Then there is a plane Hamiltonian path $\pi$ in $S$ with endpoints $s$ and $t$ such that $\overline{ab} \in \pi$, unless the quintuple $\langle S,a,b,s,t \rangle $ satisfies the conditions of \cref{lem:stisabobstruction}, \cref{lem:diagonalobstruction}, \cref{lem:bottomobstruction} or \cref{lem:complicatedobstruction}.
    We can decide if $\pi$ exists, and construct it if it does, in $O(n\log n)$ time, where $n=|S|$.
\end{restatable}

\section{Finding two plane Hamiltonian paths}\label{sec:twopaths}

In this section, we consider the following problem: We are given a point set $S$ and two distinct points $a,b\in S$. Our goal is to find two plane Hamiltonian paths $\pi_1$ and $\pi_2$ in $S$. Additionally, we require that neither, one or both of $\pi_1,\pi_2$ contain $\overline{ab}$ and are otherwise edge-disjoint.

We first observe that, if we want both $\pi_1$ and $\pi_2$ to avoid $\overline{ab}$, we can always find suitable paths.  

\begin{restatable}[\restateref{prop:twopathssegnotinc}]{proposition}{twopathssegnotinc}
    \label{prop:twopathssegnotinc}
    Let $S$ be a set of at least 5 points in general position in the plane and $a,b\in S$ such that $a \ne b$. Then there exist two plane Hamiltonian paths $\pi_1, \pi_2$ in $S$ such that $\pi_1 \cap \pi_2 = \emptyset$ and $\overline{ab} \not\in \pi_1 \cup \pi_2$.
    Such paths can be computed in polynomial time.
\end{restatable}

We remark that, for $n\ge 7$, the result of \cref{prop:twopathssegnotinc} follows from the main result of~\cite{KindermannKLV23}. For $n\in \{5,6\}$, the number of cases that need to be checked can be reduced using the framework of Pilz and Welzl~\cite{Pilz2017}; see \cref{app:proposition} for details.

The remaining two situations require a more involved treatment.

\subsection{Two edge-disjoint paths with one of them including a segment}
\label{sec:onePathWithEdgeOneWithout}

We consider the problem of finding two edge-disjoint plane Hamiltonian paths on a point set $S$, with the constraint that one of them needs to include a segment~$\overline{ab}$, where $a,b\in S$ are distinct. We show that this is always possible.

To prove \cref{Thm:OneEdgeOnePath},
we build on \cref{lem:pathTwoEndpoints,lem:pathGivenSegment}. 
In addition to these, we need the following well-known result by Abellanas, Garc{\'{\i}}a-L{\'{o}}pez, Hern{\'{a}}ndez-Pe{\~{n}}alver, Noy, and Ramos~\cite{AbellanasGHNR96,AbellanasGHNR99},
which can also be deduced from the earlier results by Hershberger and Suri~\cite{DBLP:conf/swat/HershbergerS90,HershbergerS92}. It was used to show the existence of two and three plane Hamiltonian paths in~\cite{AichholzerHKKLPSW17}
and~\cite{KindermannKLV23}.

\begin{restatable}[\cite{AbellanasGHNR99} \restateref{lem:pathAlternatingAB}]{lemma}{pathAlternatingAB}
    \label{lem:pathAlternatingAB}
    Let $S_{\mathrm{down}}$ and $S_{\mathrm{up}}$ be two sets of points in general position such that $|S_{\mathrm{down}}| \le |S_{\mathrm{up}}| \le |S_{\mathrm{down}}| + 1$
    and there is a line that separates $S_{\mathrm{down}}$ from $S_{\mathrm{up}}$. Further, let $z \in S_{\mathrm{up}}$ be any vertex (of two possible candidates) on $\partial \CH(S_{\mathrm{down}} \cup S_{\mathrm{up}})$ that
    is neighboring a vertex from $S_{\mathrm{down}}$ on $\partial \CH(S_{\mathrm{down}} \cup S_{\mathrm{up}})$.
    Then, there exists a plane Hamiltonian path $\pi$ on $S_{\mathrm{down}} \cup S_{\mathrm{up}}$ such that $\pi$ starts in~$z$ and
    every edge in $\pi$ has one endpoint in~$S_{\mathrm{down}}$ and one endpoint in~$S_{\mathrm{up}}$.
    Such a path can be found in $O(n \log n)$ time, where $n = |S_{\mathrm{down}}| + |S_{\mathrm{up}}|$.
\end{restatable}

We can now prove \cref{Thm:OneEdgeOnePath}. 

\begin{restatable}[\restateref{Thm:OneEdgeOnePath}]{theorem}{OneEdgeOnePath}
    \label{Thm:OneEdgeOnePath}
    Let $S$ be a set of at least 4 points in general position in the plane and $a,b\in S$ such that $a \ne b$.
    Then, there exist two plane Hamiltonian paths $\pi_1, \pi_2$ in $S$ such that $\pi_1 \cap \pi_2 = \emptyset$ and $\overline{ab} \in \pi_1$.
    Such paths can be computed in $O(n \log n)$ time, where $n=|S|$.
\end{restatable}

\begin{proof}
    Let us first discuss the small values of~$n$.
    To be able to have two edge-disjoint Hamiltonian paths, we need at least four points (otherwise, there are not enough edges between the points to host the edges of the two paths).
    For $n = 4$, all six possible edges between the points are
    needed to host two edge-disjoint Hamiltonian paths.
    The four points can have three or four points in the convex hull.
    In both cases, we can find two edge-disjoint plane Hamiltonian paths, as illustrated in \cref{fig:twopaths4vertices}.
    From now on, we assume that $n \ge 5$.
    
    We first rotate the point set~$S$ so that $\overline{ab}$ is horizontal.
    We then sort in $O(n \log n)$ time the points in~$S$
    in increasing order by y-coordinate (pairs of points with the same y-coordinate are
    sorted arbitrarily).
    Let $S_{\mathrm{down}}$ be the set of the first $\lfloor n/2 \rfloor$ points (the lower points)
    and let $S_{\mathrm{up}} = S \setminus S_{\mathrm{down}}$
    (the upper points).
    Clearly, $S_{\mathrm{down}}$ and $S_{\mathrm{up}}$ can be separated by a line.
    We refer to $S_{\mathrm{down}}$ and $S_{\mathrm{up}}$ as \emph{partitions} (of $S$). Since $n \ge 5$, we have that $|S_{\mathrm{down}}| \ge 2$, and $|S_{\mathrm{up}}| \ge 3$.
    
    By \cref{lem:pathAlternatingAB}, there is a plane Hamiltonian path $\pi$ on~$S$ such that every edge of $\pi$
    has an endpoint in~$S_{\mathrm{down}}$ and an endpoint in~$S_{\mathrm{up}}$,
    and $\pi$ starts in a point~$z\in S_{\mathrm{up}}$ contained in~$\partial \CH(S_{\mathrm{down}} \cup S_{\mathrm{up}})$.
    Such a path can be found in $O(n \log n)$ time.
    
    We say that a point $p \in S_{\mathrm{down}}$ \emph{sees} a point~$q \in S_{\mathrm{up}}$
    if and only if $\overline{pq} \cap \CH(S_{\mathrm{down}}) = p$ and $\overline{pq} \cap \CH(S_{\mathrm{up}}) = q$.
    If $p$ sees $q$, then $p$ lies on $\partial \CH(S_{\mathrm{down}})$,  $q$ lies on $\partial \CH(S_{\mathrm{up}})$, and $\overline{pq}$ intersects  $\CH(S_{\mathrm{down}})$ and $\CH(S_{\mathrm{up}})$ at a single point.
    We use the symmetric definition if $p \in S_{\mathrm{up}}$ and $q \in S_{\mathrm{down}}$.
    We consider two cases:
    
    \begin{description}
        \item[Case A: $a$ and $b$ are in distinct partitions.]
        We distinguish whether $\overline{ab} \in \pi$ or not.

        \item[Case A1: $\overline{ab} \notin \pi$.]
        Let $\pi_2 = \pi$. Notice that, since $\overline{ab}$ is horizontal, one of $a,b$ is the topmost point of~$S_{\mathrm{down}}$ and the other is the bottommost point of~$S_{\mathrm{up}}$.
        Let $\pi_1^{\mathrm{down}}$ and $\pi_1^{\mathrm{up}}$ be the paths traversing the points
        of~$S_{\mathrm{down}}$ and $S_{\mathrm{up}}$, respectively, such that vertices are traversed by increasing y-coordinate.
        Let $\pi_1$ be the concatenation of $\pi_1^{\mathrm{down}}$, $\overline{ab}$, and $\pi_1^{\mathrm{up}}$.
        Clearly, $\pi_1$ and $\pi_2$ are two plane Hamiltonian paths 
        such that $\overline{ab} \in \pi_1$ and $\pi_1 \cap \pi_2 = \emptyset$.
        Additionally, constructing $\pi_1$ takes $O(n)$ time after $S$ is sorted.

        \item[Case A2: $\overline{ab} \in \pi$.]
        Let $\pi_1 = \pi$.
        For the second path, we again aim to find two plane Hamiltonian paths
        on $S_{\mathrm{down}}$ and $S_{\mathrm{up}}$.
        The challenging part is to find a suitable ``bridge'' that connects the two paths
        and is not contained in $\pi_1$.
        To this end, we analyze the set of points that~$z$ (the starting point of $\pi_1$) sees.
        The conditions of the following cases can be checked in $O(n \log n)$ time.
        
        \item[Case A2.1: $z$ sees two points $x, y$ of $S_{\mathrm{down}}$.]
        At most one of $\{\overline{xz}, \overline{yz}\}$ can lie in~$\pi_1$; without loss of generality, let this be $\overline{xz}$.
        Let $\pi_2^{\mathrm{down}}$ and $\pi_2^{\mathrm{up}}$ be two plane Hamiltonian paths on $S_{\mathrm{down}}$ and $S_{\mathrm{up}}$
        starting in~$y$ and~$z$, respectively.
        Such paths exist thanks to \cref{lem:pathTwoEndpoints},
        and can be found in $O(n \log n)$ time.
        Let $\pi_2$ be the concatenation of $\pi_2^{\mathrm{down}}$, $\overline{yz}$, and $\pi_2^{\mathrm{up}}$.
        Clearly, $\pi_1$ and $\pi_2$ are two plane Hamiltonian paths 
        such that $\overline{ab} \in \pi_1$ and $\pi_1 \cap \pi_2 = \emptyset$.
        
        \begin{figure}[t]
            \centering
            \includegraphics{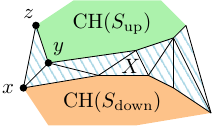}
            \caption[t]{In Case A2.2, $x$ sees at
                least two points of~$S_{\mathrm{up}}$.}
            \label{fig:x-sees-two-points}
        \end{figure}
        
        \item[Case A2.2: $z$ sees one point $x$ of~$S_{\mathrm{down}}$ and $|S_{\mathrm{down}}| = |S_{\mathrm{up}}|$.]
        We show that, if $z$ sees only one point~$x$ of~$S_{\mathrm{down}}$,
        then $x$ sees at least two points of~$S_{\mathrm{up}}$ (see \cref{fig:x-sees-two-points}):
        Let $X$ be the closure of $\CH(S) \setminus (\CH(S_{\mathrm{down}}) \cup \CH(S_{\mathrm{up}}))$.
        Clearly, $X$ is a polygonal region with vertices $x$ and $z$, and $z$ is visible to only two vertices of~$X$: $x$ and one of $z$'s neighbors on $\partial \CH(S_{\mathrm{up}})$, which we call~$y$.
        Since $X$ can be triangulated, $x$ sees $y$.
        We now perform a reflection of the point set $S$ through the line $\ell(ab)$.
        Taking the reverse order of the points sorted by y-coordinate and splitting them into two halves
        exactly switches the sets $S_{\mathrm{down}}$ and $S_{\mathrm{up}}$.
        Since $|S_{\mathrm{down}}| = |S_{\mathrm{up}}|$,
        the new point set fulfills the hypothesis of our case distinction.
        Thus, $x$ takes the role of $z$ and we are in Case~A2.1.
        
        \item[Case A2.3: $z$ sees one point of~$S_{\mathrm{down}}$ and $|S_{\mathrm{down}}| + 1 = |S_{\mathrm{up}}|$.]
        We perform a reflection of $S$ through the line $\ell(ab)$.
        Afterwards, we partition the new point set into $S_{\mathrm{up}}$ and $S_{\mathrm{down}}$ as described above. Since $|S_{\mathrm{up}}| > |S_{\mathrm{down}}|$, $a$ and $b$ end up in $S_{\mathrm{up}}$ and we are in Case~B.
    
        \item[Case B: $a$ and $b$ are in the same partition.]
        In this case, the main idea is similar as before: to use $\pi$
        as $\pi_2$ (since $\pi$ does not contain $\overline{ab}$),
        to find spanning paths of $S_{\mathrm{up}}$ and $S_{\mathrm{down}}$ including $\overline{ab}$,
        and to connect them as $\pi_1$.
        If the connection edge is easy to find,
        we are done immediately.
        This is the case, e.g., if the topmost point~$x$ of~$S_{\mathrm{down}}$ sees
        a point $s$ of $S_{\mathrm{up}}$ such that $\overline{xs} \notin \pi$
        and neither $x$ nor $s$ are identical to $a$ or $b$.
        However, there are some more subtle cases, e.g., if the candidates for the connection edges are all contained in $\pi$
        or if points like $x$ and $s$ are identical to $a$ or $b$.
        We exhaustively investigate all these special cases and their treatment in \cref{app:onePathWithEdgeOneWithout}.
        There, we show that we can always find in $O(n \log n)$ the desired paths $\pi_1$ and~$\pi_2$. \qed
    \end{description}
\end{proof}

\subsection{Two paths sharing a single segment}
\label{sec:twoPathsWithEdge}

Lastly, we consider the problem of finding two plane Hamiltonian paths on a point set $S$ that overlap on a prescribed segment $\overline{ab}$ and are otherwise disjoint.
Unlike the previous case, such paths do not always exist.
We fully characterize the triples $\langle S,a,b \rangle$ for which such paths exist.
We distinguish three cases, based on whether both, one or none of $a,b$ are vertices of $\partial\CH(S)$.

\begin{restatable}{lemma}{diagonalcase}
    \label{lem:diagonalcase}
    Let $S$ be a set of $n$ points in the plane in general position, and let $a, b \in S$ such that $a, b \in \partial \CH(S)$. There exist two plane Hamiltonian paths $\pi_1, \pi_2$ on $S$ such that $\pi_1 \cap \pi_2 = \overline{ab}$ if and only if one of the following holds:
    \begin{itemize}
        \item one of the two open halfplanes bounded by $\ell(ab)$ contains 2 or 3 points of~$S$, and the other open halfplane contains no points of~$S$; or
        \item both open halfplanes bounded by $\ell(ab)$ contain at most 1 or at least 4 points~of~$S$.
    \end{itemize}
\end{restatable}

\begin{proof}
As $a,b$ are vertices of $\partial \CH(S)$, they split $S\setminus \{a,b\}$ into
sets $S_1,S_2$; see \cref{shared-edge-1}.
Assume that $\overline{ab}$ is horizontal and $S_1$ lies above $\overline{ab}$.
We get these cases:

\begin{description}
    \item[Case A: $|S_2|=0$, $|S_1|= 1.$]
    Let $x$ be the point in $S_1.$ We define the two paths as $\pi_1 = x,a,b$ and $\pi_2 = x,b,a$. 
    
    \item[Case B: $|S_2|=0$, $|S_1|> 1.$]
    If $|S_1| = 2$ or $3$, we can construct the paths as in \cref{shared-edge-4} and \cref{shared-edge-5}.
    If $|S_1|\geq 4$, note that both $a$ and $b$ see at least $2$ vertices of $\partial\CH (S_1)$. Let $x$ be a vertex visible by $a$, and $y$ a vertex visible by $b$, with $y\neq x$. Note that, if $|S_1|=4$, we can choose $x,y$ to be two consecutive vertices of the convex hull. Then let $\pi_1', \pi_2'$ be two plane Hamiltonian paths on $S_1$ starting at $x$ and $y$, respectively, as guaranteed by \cref{Lem: KKGVThm2,Lem: S=4}. We can obtain our desired paths $\pi_1$ and $\pi_2$ by extending $\pi_1',\pi_2'$ as in \cref{shared-edge-6}.
    
    \item[Case C:  $|S_2|=1$, $|S_1|= 1.$]
    We can construct the paths as in  \cref{shared-edge-2}.

    \item[Case D: $|S_2|=1$, $|S_1| \ge 4.$]
    We find $\pi_1',\pi_2'$ as in Case C, and extend them as in \cref{shared-edge-3}.
    
    \item[Case E: $|S_2|\ge 4$, $|S_1|\ge 4$.]
    Let $x_1$ and $y_1$ (with $x_1\neq y_1$) be two vertices of $\partial\CH(S_1)$ visible from $a$ and $b$, respectively, and $x_2$ and $y_2$ be two vertices of $\partial \CH(S_2)$ defined analogously. Let $\pi_1^1,\pi_2^1$ be two plane Hamiltonian paths on $S_1$ starting at $x_1$ and $y_1$, respectively, as guaranteed by  \cref{Lem: KKGVThm2,Lem: S=4}. Analogously, let $\pi_1^2, \pi_2^2$ to be two plane Hamiltonian paths on $S_2$ starting at $x_2$ and $y_2$. Finally, we connect $\pi_1^1$ to $\pi_2^2$ and $\pi_2^1$ to $\pi_1^2$ as in \cref{shared-edge-1}, and obtain our desired paths $\pi_1$ and $\pi_2$.

\end{description}

\begin{figure}[t]
    \centering
    \begin{subfigure}[t]{0.3\textwidth}
        \centering
        \includegraphics[page=4]{Figures/shared-edge.pdf}
        \subcaption{Case B, $|S_1| = 2$.} 
        \label{shared-edge-4}
    \end{subfigure}
    \hfill
    \begin{subfigure}[t]{0.3\textwidth}
        \centering
        \includegraphics[page=5]{Figures/shared-edge.pdf}
        \subcaption{Case B, $|S_1| = 3$.} 
        \label{shared-edge-5}
    \end{subfigure}
    \hfill
    \begin{subfigure}[t]{0.3\textwidth}
        \centering
        \includegraphics[page=6]{Figures/shared-edge.pdf}
        \subcaption{Case B, $|S_1| \geq 4$.} 
        \label{shared-edge-6}
    \end{subfigure}
    
    \bigskip
    
    \begin{subfigure}[t]{0.3\textwidth}
        \centering
        \includegraphics[page=2]{Figures/shared-edge.pdf}
        \subcaption{Case C.} 
        \label{shared-edge-2}
    \end{subfigure}
    \hfill
    \begin{subfigure}[t]{0.3\textwidth}
        \centering
        \includegraphics[page=3]{Figures/shared-edge.pdf}
        \subcaption{Case D.} 
        \label{shared-edge-3}
    \end{subfigure}
    \hfill
    \begin{subfigure}[t]{0.3\textwidth}
        \centering
        \includegraphics[page=1]{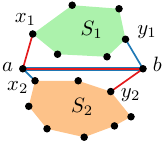}
        \subcaption{Case E.} 
        \label{shared-edge-1}
    \end{subfigure}
    \caption{Construction of $\pi_1$ and $\pi_2$ when both $a$ and $b$ are vertices of $\partial \CH(S)$.}
    \label{Fig:shared-edge}
\end{figure}

In the remaining cases it is impossible to find $\pi_1$ and $\pi_2$.

\begin{description}
    \item[Case F: $|S_2| =2$, $|S_1|\ge 1$.]
    Let $S_2 = \{x,y\}$. Assume that there are two paths $\pi_1,\pi_2$ with the desired properties. Then at most one of them can contain the segment $\overline{xy}$. Assume that $\pi_2$ does not contain this segment. Then $\pi_2$ must contain edges $\overline{ax}$ and $\overline{by}$. However, it is then clear that we cannot continue $\pi_2$ as either $a$ or $b$ would have degree $3$ or the edge $\overline{ab}$ would be crossed in $\pi_2$. Therefore, it is impossible to find such two paths.
    
    \item[Case G: $|S_2| =3$, $|S_1|\ge 1$.]
    Let $S_2 = \{x,y,z\}$. Assume that there are two paths $\pi_1,\pi_2$ with the desired properties. We can assume that $\pi_2$ contains at most one of the segments $\overline{xy},\overline{xz}, \overline{yz}$. Then, by a similar argument as in the previous case, we can see that it is impossible to complete $\pi_2$, contradicting the assumption that two such paths exist. \qed 
\end{description}
\end{proof}

The proofs of \cref{lem:one-sidedcase,lem:two-sidedcase} are moved to \cref{app:twoPathsWithEdge}.

\begin{restatable}[\restateref{lem:one-sidedcase}]{lemma}{onesidedcase}
    \label{lem:one-sidedcase}
    Let $S$ be a set of $n$ points in the plane in general position, and let $a, b \in S$ such that $a \not\in \partial \CH(S)$ and $b \in \partial \CH(S)$. There exist two plane Hamiltonian paths $\pi_1, \pi_2$ on $S$ such that $\pi_1 \cap \pi_2 = \overline{ab}$ if and only if one of the following holds:
    \begin{itemize}
        \item there are at least two bridges over $\ell(ab)$; or
        \item there is only one bridge over $\ell(ab)$ and at least three additional points of $S$ on both sides of $\ell(ab)$.
    \end{itemize}
\end{restatable}

\begin{restatable}[\restateref{lem:two-sidedcase}]{lemma}{twosidedcase}
    \label{lem:two-sidedcase}
    Let $S$ be a set of $n$ points in the plane in general position, and let $a, b \in S$ be two points not lying on the convex hull of $S$. There exist two plane Hamiltonian paths $\pi_1, \pi_2$ on $S$ such that $\pi_1 \cap \pi_2 = \overline{ab}$.
\end{restatable}

Now, we can state the main result of this section.

\begin{restatable}{theorem}{OneEdgeBothPaths}
    \label{thm:OneEdgeBothPaths}
    Let $S$ be a set of at least 4 points in the plane in general position, and let $a,b\in S$ such that $a \ne b$. There exist two plane Hamiltonian paths $\pi_1, \pi_2$ on $S$ such that $\pi_1 \cap \pi_2 = \overline{ab}$ if and only if the triple $\langle S,a,b \rangle$ satisfies one of the \cref{lem:diagonalcase,lem:one-sidedcase,lem:two-sidedcase}. Moreover, we can decide if $\pi_1,\pi_2$ exist, and construct them if they do, in $O(n\log n)$ time, where $n=|S|$.
\end{restatable}

\begin{proof}
    The existence part of the statement follows from \cref{lem:diagonalcase,lem:one-sidedcase,lem:two-sidedcase}. Note that the construction of the paths $\pi_1,\pi_2$ in the proofs of \cref{lem:diagonalcase,lem:one-sidedcase,lem:two-sidedcase} is based on multiple applications of \cref{Lem: KKGVThm2,Lem: S=4}, and ad-hoc constructions for small point sets. Thus, the running time is $O(n\log n)$.
\end{proof}

\section{Open problems}\label{sec:open}
We conclude by listing some open problems.
Firstly, a natural way to unify and extend \cref{Thm:OneEdgeOnePath,thm:stpathabincluded} would be to determine under what conditions two paths can be found when one or both endpoints are prescribed to one or both of the paths. Doing this would probably require significant refinements of the arguments presented in the proof of \cref{Thm:OneEdgeOnePath}. 
A somewhat simpler problem might be to extend \cref{prop:twopathssegnotinc} and \cref{thm:OneEdgeBothPaths} in a similar way by prescribing one or both endpoints to one or both paths.
Another way to generalize our results is by prescribing more than one edge to be taken or avoided by one or two paths.

Lastly, the most challenging open problem
is to reduce the gap between the upper and lower bounds for the geometric graph packing problem for plane Hamiltonian paths. The best known lower bound is $3$ and the best upper bound is $\lceil\frac{n}{3}\rceil$, both proven in~\cite{KindermannKLV23}.

\subsubsection{\ackname} We thank the organizers of the HOMONOLO 2024 workshop in Nov\'a Louka, Czech Republic, for the fruitful atmosphere where the research on this project was initiated. \\
T. Anti\'{c}, A. D\v{z}uklevski, J. Kratochv\'{i}l and M. Saumell received funding from GA\v{C}R grant 23-04949X, T.A and A.D\v{z} were additionally supported by GAUK grant VV–2025–260822. G. Liotta was supported in part by  MUR of Italy, PRIN Project no. 2022TS4Y3N – EXPAND and PON Project ARS01\_00540

\subsubsection{\discintname} The authors have no competing interests to declare that are relevant to the content of this article.

\bibliographystyle{splncs04}
\bibliography{bibliography,biblio} 

\clearpage
\appendix

\section{Omitted proof of \cref{sec:prelims}}

\pathGivenSegment* \label{lem:pathGivenSegment*}

\begin{proof}
    Rotate the point set~$S$ such that $\overline{ab}$ is a horizontal line segment.
    Without loss of generality, the y-coordinate of $z$ is smaller than that of~$a$ and~$b$.
    We will obtain $\pi$ by taking a plane Hamiltonian path through the points below $\overline{ab}$,
    attaching $\overline{ab}$, and adding a plane Hamiltonian path through the points above $\overline{ab}$;
    see \cref{fig:segment-and-startpoint-given}.
    
    \begin{figure}[t]
        \centering
        \includegraphics{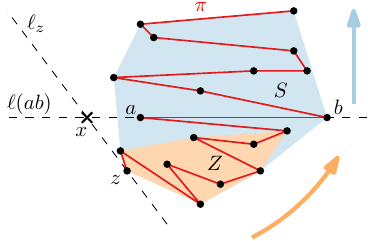}
        \caption[t]{Our strategy in the proof of \cref{lem:pathGivenSegment}:
            the path $\pi$ starts in~$z$, traverses the
            points in~$Z$ radially around~$x$, contains~$\overline{ab}$,
            and traverses the remaining points by y-coordinate.}
        \label{fig:segment-and-startpoint-given}
    \end{figure}
    
    Let $Z$ be the set of points in~$S$ that have a y-coordinate strictly smaller than that of~$a$ and~$b$.
    Let $\ell_z$ be a line through $z$ such that $\ell_z \cap \CH(S) = z$
    and $\ell_z$ is not parallel to~$\ell(ab)$.
    Such a line exists because $z$ lies on the convex hull of~$S$
    and $S$ is in general position.
    Without loss of generality, let the intersection point~$x$ of $\ell_z$ and $\ell(ab)$
    be to the left of~$\overline{ab}$.
    The first part of~$\pi$ collects the points of~$Z$
    sorted by their radial coordinates around~$x$
    in such a way that~$z$ is the first point and the other points of~$Z$
    follow in counter-clockwise order around~$x$.
    Then, we add $a$ and $b$ to~$\pi$.
    Finally, we add to~$\pi$ the points of $S \setminus (Z \cup \{a, b\})$
    by increasing y-coordinate.
    Clearly, $\pi$ is a plane Hamiltonian path that includes $\overline{ab}$ and starts in~$z$.
    
    An algorithm realizing this procedure is straight-forward.
    It needs to sort the points by radial coordinates and by y-coordinate
    and, hence, can run in $O(n \log n)$ time.
\end{proof}

\section{Omitted proofs of \cref{sec:proofofstpathsegnotinc}}
\label{app:proofofstpathsegnotinc}

\convexobstacle* \label{lem:convexobstacle*}

\begin{proof}
    Since $S$ is in convex position and $s,t$ are adjacent on $\partial \CH(S)$, there is only one Hamiltonian $s$--$t$ path in $S$ and it uses all edges of the convex hull, including~$\overline{ab}$.
\end{proof}

\wheelobstacle* \label{lem:wheelobstacle*}

\begin{proof}
    Assume, for the sake of contradiction, that the desired path $\pi$ exists. Then it contains a subpath $\pi'$ that connects $a$ and $b$. Note that this path contains at least one diagonal of $\partial\CH(S)$ connecting a vertex of $\partial\CH(S)$ lying between $a$ and $y$ (this vertex might be $a$ or $y$) to a vertex of $\partial \CH(S)$ lying between $b$ and $x$ (this vertex might be $b$ or $x$). However, by our assumption, each such diagonal either separates both $a,b$ from $t$ (this happens when neither of $a,b$ is a vertex of the diagonal) or separates $s,t$ from one of the points $a,b$ (for example if $b$ is vertex of the diagonal but $a$ is not). It is important to note that since we require that $\overline{ab}\not\in\pi$, it is impossible that this diagonal is $\overline{ab}$.  In each case, it is impossible to complete $\pi'$ to a plane Hamiltonian $s$--$t$ path. 
\end{proof}

\oneendptsegnotinc*
\label{lem:oneendptsegnotinc*}

\begin{proof}    
    We split the proof into three cases.
    
    \begin{figure}[t]
        \centering
        
        \begin{subfigure}[t]{0.3\textwidth}
            \includegraphics[page=1]{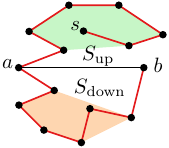}
            \subcaption{Case B.}
            \label{fig:oneendptsegnotinc-1}
        \end{subfigure}
        \hfill
        \begin{subfigure}[t]{0.3\textwidth}
            \includegraphics[page=2]{Figures/lemma10.pdf}
            \subcaption{Case C1 -- the construction when $s\not\in \triangle(axb)$.}
            \label{fig:oneendptsegnotinc-2}
        \end{subfigure}
        \hfill
        \begin{subfigure}[t]{0.3\textwidth}
            \includegraphics[page=3]{Figures/lemma10.pdf}
            \subcaption{Case C1 -- the construction when $s\in \triangle(a,x,b)$.}
            \label{fig:oneendptsegnotinc-3}
        \end{subfigure}
        
        \bigskip
        
        \begin{subfigure}[t]{0.3\textwidth}
            \includegraphics[page=5]{Figures/lemma10.pdf}
            \subcaption{Case C2 -- the construction when $s \in \partial\CH(S)$}
            \label{fig:oneendptsegnotinc-4}
        \end{subfigure}
        \hfil
        \begin{subfigure}[t]{0.3\textwidth}
            \includegraphics[page=4]{Figures/lemma10.pdf}
            \subcaption{Case C2 -- the construction when $s\not\in \partial\CH(S)$.}
            \label{fig:oneendptsegnotinc-5}
        \end{subfigure}
        \caption{Construction of path $\pi$ in different cases of \cref{lem:oneendptsegnotinc}.}
        \label{fig:oneendptsegnotinc}
    \end{figure}
    
    \begin{description}
        \item[Case A: $s \in \{a, b\}$.]
        Assume $s = a$.
        Apply \cref{lem:pathTwoEndpoints} to obtain an $a$--$b$ path in $S$.
        Since $n \ge 4$, the edge $\overline{ab}$ is not used in $\pi$.

        \item[Case B: $s \not\in \{a,b\}$ and $\overline{ab}$ is not an edge of $\CH(S)$.]
        Assume that $\overline{ab}$ is horizontal and that $s$ lies above $\overline{ab}$. Then, we apply \cref{lem:pathTwoEndpoints} to the point set consisting of all points of $S$ lying above $\overline{ab}$ plus the point $a$ to find a plane $s$--$a$ path. We then again apply \cref{lem:pathTwoEndpoints} to the remaining points plus the point $a$ to find a plane $a$--$b$ path, which cannot include the segment $\overline{ab}$ since our assumptions ensure that there is at least one point below $\overline{ab}$. See \cref{fig:oneendptsegnotinc-1}.
        
        \item[Case C:  $s \not\in \{a,b\}$ and $\overline{ab}$ is an edge of $\CH(S)$.]
        Assume again that $\overline{ab}$ is horizontal and that all points of $S$ lie above it. Now we split into two subcases.
        
        \item[Case C1: $\CH(S)$ is not a triangle.] In this case, let $x$ be the vertex of $\partial\CH(S)$ different from $b$ that is adjacent to $a$. We can assume that $x\neq s$ (otherwise, we exchange $a$ and $b$). If $s\not \in \triangle(a,x,b)$, we apply \cref{lem:pathTwoEndpoints} to the set $S\setminus \triangle(a,x,b) \cup \{b\}$ to obtain a plane $s$--$b$ path and once more to the set $S\cap \triangle
        (a,x,b)$ to obtain a plane $b$--$a$ path. Concatenating these two paths gives us $\pi$.  See \cref{fig:oneendptsegnotinc-2}.
        Otherwise, if $s\in \triangle (a,x,b)$, we first apply \cref{lem:pathTwoEndpoints} to obtain an $s$--$x$ path in $\triangle(a,x,b)\setminus\{b\}$, and concatenate it with any $x$--$b$ path on the remaining points. See \cref{fig:oneendptsegnotinc-3}.
        \item[Case C2: $\CH(S)$ is a triangle] 
        If $s$ is the third vertex of $\partial\CH(S)$ we obtain $\pi$ by applying \cref{lem:pathTwoEndpoints} to the set $S \setminus \{s\}$ to obtain an $a$--$b$ path and then concatenate it with the segment $\overline{as}$, see \cref{fig:oneendptsegnotinc-4}. Otherwise if $s\not\in\partial\CH(S)$, we construct $\pi$ by concatenating the path from $a$ to $b$ on $\partial\CH(S)$ with any $b$--$s$ path which uses only interior points of $S$ - we find this path by applying \cref{lem:pathTwoEndpoints} to the set of interior points of $S$ together with $b$. See \cref{fig:oneendptsegnotinc-5}.
    \end{description}
    Applying \cref{lem:pathTwoEndpoints} and doing local modifications
    allows for a running time in $O(n \log n)$.
\end{proof}

\stpathsegnotinc*
\label{thm:stpathsegnotinc*}

\begin{proof}
    We remark that cases A and E1 from the main body of the proof have been completely proven. Therefore, to finish the proof, we only need to show that the desired path $\pi$ exists in the following cases:
    
    \begin{description}
        \item[Case B: $s$ and $t$ lie in the same closed halfplane determined by $\ell(ab)$] \phantom{breakafterthis}
        \textbf{and $\overline{ab}$ is not an edge of $\partial\CH(S)$.}
        ~We assume that $s,t$, lie on or below $\ell(ab)$. First, we find an $a$--$b$ path in the point set containing all the points above $\ell(ab)$ and including $a,b$. Next, we find a line $p$ that separates $s,t$ so that, up to renaming, $a$ and $s$  lie in the same halfplane determined by $p$ while $b$ and $t$ lie in the complementary halfplane. We find an $a$--$s$ and a $b$--$t$ path and concatenate all three of them. Note that it may happen that $s \in \{a,b\}$ or $t\in\{a,b\}$. In this case we find a single path path collecting all points below $\ell(ab)$ (for example if $s=a$, such a path would be a $b$--$t$ path). The existence of  all considered paths is guaranteed by \cref{lem:pathTwoEndpoints}. 
        
        From now on, we may assume that $\overline{ab} \in \partial \CH(S)$ and split into the following cases: 
        
        \item[Case C: $s$ and $t$ lie in the same closed halfplane determined by $\ell(ab)$,]
        \textbf{$S$ is in convex position and $\overline{st}$ is not an edge of $\partial\CH(S)$.}
        ~Recall that this is the only case where $S$ is in convex position that we have to consider, as otherwise, we know that the desired path $\pi$ does not exist by \cref{lem:convexobstacle}. Assume that $a,b,s,t$ appear in this order when traversing $\partial \CH(S)$. Note that there is at least one point between $s,t$ on $\partial \CH(S)$; note also that it is possible that $s\in \{a,b\}$ or $t\in \{a,b\}$. Then we obtain $\pi$ by concatenating the path connecting $s$ to $b$ using convex hull edges but not including $a$, the path connecting $t$ to $a$ using convex hull edges but not including $b$, and the only possible path connecting $a$ to $b$ avoiding $\overline{ab}$ using the remaining points. See \cref{fig:onepathsegnotinc-2}. Note that if $s \in \{a,b\}$ or $t\in\{a,b\}$, then one of these paths may in fact be of length zero.
        
        \item[Case D:  $s$ and $t$ lie in the same closed halfplane determined by $\ell(ab)$] \phantom{breakafterthis}
        \textbf{and $S$ contains only one point in the interior of $\CH(S)$.}
        ~If $S$ contains only one point $p$ in the interior of $\CH(S)$ and $p \notin \{s,t\}$, it is easy to see that $\pi$ exists (just take any plane Hamiltonian $s$--$t$ path on $S\setminus\{p\}$ which contains $\overline{ab}$ and then replace $\overline{ab}$ by a path of length two using $\overline{ap},\overline{pb}$). If $p=s$ and $a=t$ or $b=t$ again finding a path is easy (if $a=t$, start with segment $\overline{pb}$ and then go from $b$ to $a$ by traversing the convex hull).
        
        So we can assume that $p=s$ and that  $t\not\in \{a,b\}$. Now let $x,y$ be two neighbours of $t$ such that $a,b,x,t,y$ appear on $\partial\CH(S)$ in this order (as in \cref{lem:wheelobstacle}). Since we assume that \cref{lem:wheelobstacle} does not apply to $\langle S,a,b,s,t\rangle$, we can assume that $\ell(by)$ separates $s$ and $t$ -- note that this assumption guarantees $a\neq y$. 
        We construct $\pi$ by starting with the segment $\overline{as}$, then following a path from $a$ to $y$ on $\partial\CH(S)$, using the segment $\overline{yb}$ and lastly following a path from $b$ to $t$ on $
        \partial\CH(S)$. See \cref{fig:onepathsegnotinc-3}.
        
        \item[Case E2: $s$ is the only point such that $\triangle(a,s,b)$ is empty.]
        In this case, all interior points of $S$ are contained in a single halfplane determined by $\ell(as)$ and in a single halfplane determined by $\ell(bs)$. We now again consider two situations. If $t\not \in \partial\CH(S)$, we construct $\pi$ by first traversing $\partial\CH(S)$ from $a$ to $b$, and concatenating this path with segment $\overline{as}$ and any $b$--$t$ path on the remaining points, which we can obtain by applying \cref{lem:pathTwoEndpoints}; see  \cref{fig:onepathsegnotinc-5}. 
        If $t\in \partial\CH(S)$, we first traverse $\partial\CH(S)$ from $t$ to $a$ (including $a$), in the direction in which we do not use $\overline{ab}$. We then collect a neighbor $p\neq s,b$ of $a$ on the convex hull of the  remaining points  (note that such a $p$ exists since we assume that there are at least two points in the interior of $\CH(S)$). Call thus obtained $t$--$p$ subpath $\pi'$ and the set of points on which it is defined $S'$. Now we construct $\pi$ by taking the segment $\overline{sb}$, concatenating it with any $b$--$p$ path in $(S\setminus S') \cup \{p\}$ (which we know it exists by \cref{lem:pathTwoEndpoints}), and finally with $\pi'$; see \cref{fig:onepathsegnotinc-6}.
    \end{description}
    
    This concludes the proof of \cref{thm:stpathsegnotinc}.
\end{proof}

\section{Omitted proofs of \cref{sec:proofofstpathabincluded}}
\label{app:proofofstpathabincluded}

\diagonalobstruction*
\label{lem:diagonalobstruction*}
\begin{proof}
    Assume on the contrary that such a path exists. Since $\overline{ab}$ is not an edge of $\partial\CH(S)$, then there is at least one point $p\in S$ lying on the opposite side of $\ell(ab)$ from $s,t$. Then, any plane Hamiltonian $s$--$t$ path needs to visit $p$ between visiting  $a$ and $b$ and therefore cannot include $\overline{ab}$.  
\end{proof}

\bottomobstruction*
\label{lem:bottomobstruction*}
\begin{proof}
    Without loss of generality, assume that $a=s$ and that $t$ is below $\ell(ab)$. We use sets $S_{\mathrm{up}}$ and $S_{\mathrm{down}}$, defined as above. Since both $S_{\mathrm{up}}$ and $S_{\mathrm{down}}\!\setminus t$ are nonempty, 
    any Hamiltonian $s$--$t$ path containing $\overline{ab}$ contains an edge between these two sets. 
    By the assumptions, each segment connecting a point from $S_{\mathrm{up}}$ and $S_{\mathrm{down}}\!\setminus t$ intersects $\overline{ab}$, hence the path is not plane as required.
\end{proof}

\complicatedobstruction*
\label{lem:complicatedobstruction*}

\begin{proof}
    Assume that $s \in \{a,b\}$ and that $t$ lies below $\ell(ab)$. Then we know that~$t$ is a vertex of $\partial\CH(S)$ which is closest to $\ell(ab)$ among all vertices of $\partial\CH(S)$ lying below $\ell(ab)$ (and possibly the only one). Further, we know that for every bridge $\overline{xy}$ over $\ell(ab)$ such that $t\not\in \{x,y\}$, it holds that both $x,y$ are vertices of $\partial\CH(S)$ and both open halfplanes determined by $\ell(xy)$ contain at least one point above and below $\ell(ab)$. Now assume that a plane Hamiltonian $s$--$t$ path $\pi$ exists in $S$.
    Then, the segment in $\pi$ which contains $t$ is either a bridge over $\ell(ab)$, or it connects $t$ to another point lying below $\ell(ab)$. If the former is the case, then~$\pi$ starts at $t$, collects some points of $S$ lying above $\ell(ab)$ and then needs to use another bridge $\overline{xy}$ over $\ell(ab)$, but any such bridge separates the remaining points into two nonempty parts and therefore $\pi$ would either not be plane or spanning. In the other case the situation is almost identical, except that $\pi$ first collects points lying below $\ell(ab)$ and then uses a bridge $\overline{xy}$. See \cref{fig:negative_case_lem16}.
\end{proof}

\begin{figure}[t]
    \centering
    \includegraphics[page=1]{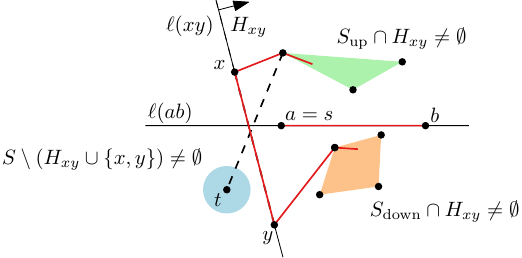}
    \caption[t]{The case described by \cref{lem:complicatedobstruction}.}
    \label{fig:negative_case_lem16}
\end{figure}

\stpathabincluded*
\label{thm:stpathabincluded*}

\begin{proof}    
    For $|S| = 2$, the edge $\overline{ab}=\overline{st}$ is the desired solution.
    Therefore, we can assume that $|S|> 2$.
    If $\{s,t\}=\{a,b\}$, it is impossible to find the desired path~$\pi$ by \cref{lem:stisabobstruction}.
    From now on, we assume that $\{s,t\}\neq \{a,b\}$ and that $\overline{ab}$ is horizontal with $a$ to the left of $b$.
    
    If $s$ and $t$ lie in different closed halfplanes determined by the line $\ell(ab)$, then
    the assumptions guarantee that $s,t\not\in\{a,b\}$. Assume that $s$ is above and $t$ is below $\ell(ab)$. We split $S$ into two subsets, $S_{\mathrm{up}}$ and $S_{\mathrm{down}}$, where $S_{\mathrm{up}}$ consists of all points of $S$ above $\ell(ab)$, and $S_{\mathrm{down}} = S \setminus (S_{\mathrm{up}} \cup \{a,b\})$ are the points below this line.
    Then, by \cref{lem:pathTwoEndpoints}, we can find a plane Hamiltonian $s$--$a$ path in $S_{\mathrm{up}}\cup \{a\}$ and a plane Hamiltonian $b$--$t$ path in $S_{\mathrm{down}}\cup \{b\}$. Joining these paths using the segment $\overline{ab}$ gives us the path $\pi$. See \cref{fig:stpathabincluded-1}.

    \begin{figure}[t]
        \centering
        \begin{subfigure}[t]{.32 \textwidth}
            \centering
            \includegraphics[page=1]{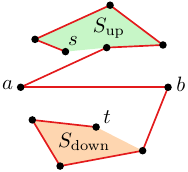}
            \subcaption{$\ell(ab)$ separates $s$ from~$t$.}
            \label{fig:stpathabincluded-1}
        \end{subfigure}
        \hfill
        \begin{subfigure}[t]{.3 \textwidth}
            \centering
            \includegraphics[page=2]{Figures/stpathsegnotinc_new.pdf}
            \subcaption{Case A.}
            \label{fig:stpathabincluded-2}
        \end{subfigure}
        \hfill
        \begin{subfigure}[t]{.3 \textwidth}
            \centering
            \includegraphics[page=3]{Figures/stpathsegnotinc_new.pdf}
            \subcaption{Case B.1.}
            \label{fig:stpathabincluded-3}
        \end{subfigure}
        \caption{Construction of path $\pi$ in different cases of \cref{thm:stpathabincluded}.}
        
        \label{fig:stpathabincluded}
    \end{figure}
    
    From now on, we assume that $s$ and $t$ lie in the same closed halfplane determined by the line $\ell(ab)$, say below or on $\ell(ab)$. Now we split into the following cases: 
    
    \begin{description}
        \item[Case A: $\overline{ab}$ is an edge of $\CH(S)$.]
        If some of $s,t$ coincides with $a$ or $b$, say w.l.o.g. $s = a$, then we find a plane Hamiltonian $b$--$t$ path in $S \setminus \{s\}$ by \cref{lem:pathTwoEndpoints} and obtain $\pi$ by attaching this path to the edge $sb$. So we can assume that $\{s,t\}\cap \{a,b\} = \emptyset$. We suitably match points from these two sets, say w.l.o.g.\ $a$ with $s$ and $b$ with $t$, so that there exists a line $p$ such that the pairs  $a,s$ and $b,t$ lie in distinct halfplanes determined by $p$. Again, we find an $s$--$a$ and a $b$--$t$ path in these two halfplanes by \cref{lem:pathTwoEndpoints}, and obtain $\pi$ by connecting them using the segment $\overline{ab}$. See \cref{fig:stpathabincluded-2}.
        
        \item[Case B: $\overline{ab}$ is not an edge of $\CH(S)$.]
        Since the case $\{a,b\}\subset \CH(S)$ is excluded by \cref{lem:diagonalobstruction}, we may assume that $\{a,b\}\not\subset \CH(S)$ and at least one bridge over $\ell(ab)$ exists.
        We differentiate between two situations:
        
        \item[Case B1: $\{a,b\} \cap \{s,t\} = \emptyset$.]
        Let $\overline{xy}$ be a bridge over $\ell(ab)$ such that $x$ lies above $\ell(ab)$ and $\overline{xy}$ is an edge of $\partial \CH(S)$. We can assume that $\overline{xy}$ crosses $\ell(ab)$ to the left of $a$. If $y\in \{s,t\}$ (assume $y=s$), we then obtain $\pi$ by concatenating $\overline{yx}$ with any plane $x$--$a$ path on the points lying above $\ell(ab)$ (including $a$ but not $b$), segment $\overline{ab} $ and lastly with any plane $b$--$t$ path on the remaining points, where these paths are obtained by applications of \cref{lem:pathTwoEndpoints}. Otherwise, assume $y\not\in\{s,t\}$ and recall that $S_{\mathrm{up}}$ denotes the set of all points in $S$ lying above $\ell(ab)$. Further, let $p$ be a line such that, up to renaming of $s$ and $t$, the pairs $a,s$ and $b,t$ lie in different halfplanes determined by $p$, as in Case A.  We then denote by $S_s$ the set of points lying below $\ell(ab)$ and in the same halfplane determined by $p$ as $s$, and by $S_t$ the set of points lying below $\ell(ab)$ and in the same halfplane determined by $p$ as $t$. Then we can see that $S = S_{\mathrm{up}} \cup S_s \cup S_t \cup \{a\} \cup \{b\}$. Now by \cref{lem:pathTwoEndpoints} we can find a $t$--$b$ path in $S_t\cup \{b\}$, an $a$--$x$ path in $S_{\mathrm{up}}\cup \{a\}$ and a $y$--$s$ path in $S_s$. Finally, we obtain $\pi$ by concatenating these paths by segments $\overline{ab}$ and $\overline{xy}$ in a suitable order. See \cref{fig:stpathabincluded-3}.
        
        \item[Case B2: $\{a,b\} \cap \{s,t\} \ne \emptyset$.]
        
        Without loss of generality, assume that $s = a$.
        If there is a bridge $\overline{xy}\in\partial\CH(S)$ over $\ell(ab)$ not containing $t$, say with $x\in S_{\mathrm{up}}$, then compose $\pi$ from $\overline{ab}$ followed by a $b$--$x$ path in $S_{\mathrm{up}}\cup \{b\}$, then by $\overline{xy}$ and finally by a $y$--$t$ path in $S_{\mathrm{down}}$. The existence of both paths is guaranteed by \cref{lem:pathTwoEndpoints}; see \cref{fig:stpathsegnotinc-4}.

        From now on, we assume that all (one or both) bridges over~$\ell(ab)$ which are also edges of~$\partial\CH(S)$ contain~$t$.
        
        As in the forthcoming analysis of subcases
        we involve similar arguments to compose~$\pi$ from parts, we introduce a notation in which the constructed path~$\pi$ will be encoded as  
        $\pi=\overline{ab}\circ (b,S_{\mathrm{up}}\cup\{b\},x)\circ \overline{xy}\circ (y,S_{\mathrm{down}},t)$. Here~$\circ$ stands for concatenation, and each triple of the form $(s',S',t')$ means a path from $s'$ to $t'$ in $S'$ (guaranteed by \cref{lem:pathTwoEndpoints}), or just an empty path when $S'=\{s'\}=\{t'\}$.
        For correctness we require $s',t'\in \partial\CH(S')$, besides further arguments showing that no crossings between the concatenated parts appear.
        
         If $S_{\mathrm{down}}=\{t\}$ and $\overline{xt}$ is a bridge, then we may choose 
        $\pi=\overline{sb}\circ (b,S_{\mathrm{up}}\cup\{b\},x)\circ \overline{xt}$.
        
        Suppose next that $S_{\mathrm{down}}\!\setminus t\ne\emptyset$. Recall that, by \cref{lem:bottomobstruction}, no path $\pi$ exists when $S_{\mathrm{down}}\!\setminus t\ne\emptyset$ and all bridges over $\ell(ab)$ contain $t$.    
        So we may assume that a bridge $\overline{xy}$ over $\ell(ab)$ with $x\in S_{\mathrm{up}}$ and $y\in S_{\mathrm{down}}\!\setminus t$ exists.
        
        \begin{figure}[t]
            
            \centering
            \begin{subfigure}[t]{.32 \textwidth}
                \centering
                \includegraphics[page=4]{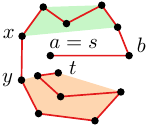}
                \subcaption{A bridge without $t$ exists.}
                \label{fig:stpathsegnotinc-4}
            \end{subfigure}   
            \hfil
            \begin{subfigure}[t]{0.29\textwidth}
                \centering
                \includegraphics[page=5]{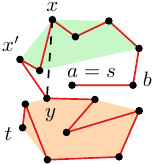}
                \subcaption{Case B2.1a.}
                \label{fig:stpathsegnotinc-5}
            \end{subfigure}
            \hfil
            \begin{subfigure}[t]{0.37\textwidth}
                \centering
                \includegraphics[page=6]{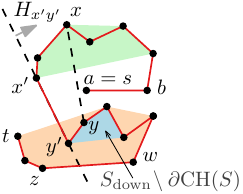}
                \subcaption{Case B2.1b.}
                \label{fig:stpathsegnotinc-6}
            \end{subfigure}
            
            \begin{subfigure}[t]{0.3\textwidth}
                \centering
                \includegraphics[page=7]{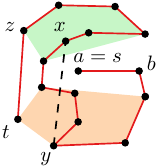}
                \subcaption{Case B2.2.}
                \label{fig:stpathsegnotinc-7}
            \end{subfigure}
            \hfil
            \begin{subfigure}[t]{0.3\textwidth}
                \centering
                \includegraphics[page=8]{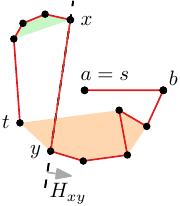}
                \subcaption{Case B2.3a.}
                \label{fig:stpathsegnotinc-8}
            \end{subfigure}
            
            \caption{Construction of path $\pi$ in \cref{thm:stpathabincluded}.}
            \label{fig:stpathsegnotinc}
        \end{figure}
        
        \item[Case B2.1: There exists a bridge with $y\notin\partial\CH(S)$.]
        
        \item[Case B2.1a: There exists a bridge with $y\in\partial\CH(S_\mathrm{down})\setminus\partial\CH(S)$.]
        We \phantom{makelinebreak} choose $x'$ such that $\overline{x'y}$ forms a bridge over $\ell(ab)$ and also $\overline{x'y}\in \partial\CH(S_\mathrm{up}\cup\{y\})$. Such $x'$ exists, because the polygon $\partial\CH(S_\mathrm{up}\cup\{y\})$ may intersect $\overline{ab}$ at most once, because otherwise $y$ could not participate in any bridge over $\ell(ab)$.
        We then get $\pi=\overline{sb}\circ (b,S_{\mathrm{up}},x')\circ \overline{x'y}\circ (y,S_{\mathrm{down}}
        ,t)$.
        See \cref{fig:stpathsegnotinc-5}.
        
        \item[Case B2.1b: There exists no bridge with $y\in\partial\CH(S_\mathrm{down})\setminus\partial\CH(S)$.]

        In \phantom{makelinebreak} this case we aim to compose $\pi$ so that after traversing a bridge over $\ell(ab)$  from $S_{\mathrm{up}}$ to $S_{\mathrm{down}}$ it first includes all points of a suitable convex subset of $S_{\mathrm{down}}$ and then it reaches $t$ along $\partial \CH(S)$.

        For this purpose we determine a bridge $x'y'$ over $\ell(ab)$ for the 
        set $S\setminus(\partial \CH(S) \cap S_{\mathrm{down}})$ on the boundary of this set (either one or two such exist). Formally, we choose it so that 
        $x'=\overline{x'y'}\cap \partial\CH(S_{\mathrm{up}})$, $y'=\overline{x'y'}\cap \partial\CH(S_{\mathrm{down}}\!\setminus\partial\CH(S))$ and  $\overline{x'y'}\in \partial\CH(S\setminus (\partial \CH(S)\cap S_{\mathrm{down}}))$.
        
        Let $H_{x'y'}$ denote the open halfplane determined by the line $\ell(x'y')$ and the point~$a$.
        Let $\overline{wz}$ be the segment in which the line $\ell(x'y')$ intersects $\partial\CH(S)\cap S_{\mathrm{down}}$, say with $w\in H_{x'y'}$. Then $y',w\in \partial\CH(S_{\mathrm{down}}\cap H_{x'y'})$ and segments $\overline{x'y'}$ and $\overline{wz}$ do not intersect.
        
        Then we define: $\pi=\overline{sb}\circ (b,S_{\mathrm{up}},x')\circ \overline{x'y'}
        \circ (y',S_{\mathrm{down}}\cap H_{x'y'},w)\circ \overline{wz} \circ (z,S_{\mathrm{down}}\!\setminus H_{x'y'},t)$.  
        See \cref{fig:stpathsegnotinc-6}.
        
        \item[Case B2.2: There exists a bridge with $x\notin\partial\CH(S)$.] 
        
        In this case, let $z\in \partial\CH(S) \cap S_{\mathrm{up}}$ be the end of a bridge over $\ell(ab)$ whose other end is $t$ and $\overline{tz}$ is an edge of $\partial\CH(S)$. We first look for an $s$--$z$ path $\pi'$ in $S\setminus t$ containing $\overline{ab}$.
        
        If the set of vertices in $\partial\CH(S\setminus t)\cap S_{\mathrm{up}}$ is the same as the one in $\partial\CH(S)\cap S_{\mathrm{up}}$, we are back in Case~B2.1. To see this, note that we are looking for a $z$--$s$ path in $S\setminus t$ and that $z$ lies above $\ell(ab)$ so the fact that there is a bridge with $x\not\in\partial\CH(S)$ is symmetric to the conditions of Case~B2.1, and therefore we can find the desired path $\pi'$.
        
        Otherwise, there is some vertex of $\partial\CH(S\setminus t)\cap S_{\mathrm{up}}$ which is not a vertex of $\partial\CH(S)\cap S_{\mathrm{up}}$ and therefore $z$ is not incident to a bridge which is an edge of $\partial\CH(S\setminus t)$, which corresponds to the situation discussed in the very beginning of Case B.2 (again with $t$ replaced by $z$ and roles of $S_{\mathrm{up}}$ and $S_{\mathrm{down}}$ exchanged), and therefore we can find the desired path $\pi'$.
        In both cases, we get $\pi=\pi'\circ \overline{zt}$.  
        See \cref{fig:stpathsegnotinc-7}.
        
        \item[Case B2.3: All bridges satisfy $x,y\in\partial\CH(S)$.] 
        Let $H_{xy}$ denote the open halfplane determined by the line $\ell(xy)$ and the point $a$.
        
        \item[Case B2.3a: There exists a bridge such that $H_{xy} \cap S_{\mathrm{up}} = \emptyset$.]
        
        In this case, observe that $S_{\mathrm{down}}\cap H_{xy}$ contains points within an angle spanned by lines $\ell(ab)$ and $\ell(xy)$, so $(S_{\mathrm{down}}\cap H_{xy})\cup \{b,y\}$ is a set with $b$ and $y$ on the boundary of its convex hull. Analogously, $x$ is on $\partial(\CH( S\setminus H_{xy}\setminus\{y\}))$.
        
        We then may define:
        $\pi=\overline{sb}\circ (b,(S_{\mathrm{down}}\cap H_{xy})\cup \{b,y\},y)\circ \overline{yx}\circ (x,S\setminus H_{xy}\setminus\{y\},t)$. See \cref{fig:stpathsegnotinc-8}.
        
        \item[Case B2.3b: There exists a bridge such that $H_{xy} \cap S_{\mathrm{down}} = \emptyset$.]
        
        By a symmetric argument to the previous case with the roles of $x$ and $y$ as well as $S_{\mathrm{up}}$ and $S_{\mathrm{down}}$ exchanged we obtain:
        $\pi=\overline{sb}\circ (b,(S_{\mathrm{up}}\cap H_{xy})\cup \{b,x\},x)\circ \overline{xy}\circ (y,S\setminus H_{xy}\setminus\{x\},t)$. 
        
        \item[Case B2.3c: All bridges satisfy  $H_{xy} \cap S_{\mathrm{up}} \ne \emptyset$ and $H_{xy} \cap S_{\mathrm{down}} \ne \emptyset$.]
        
        In this case, the path does not exist by \cref{lem:complicatedobstruction}. See also \cref{fig:negative_case_lem16}
    \end{description}
    
    Applying \cref{lem:pathTwoEndpoints} and doing local modifications
    yields a running time in $O(n \log n)$.
    Observe that we can check the existence of suitable bridges in $O(n \log n)$ time by sorting the points radially around~$a$ and around~$b$.
\end{proof}

\section{Proof of \cref{prop:twopathssegnotinc}}\label{app:proposition}

We first need to briefly discuss a tool that will allow us to keep our case analysis relatively small. In particular, we use the notion of \emph{crossing-dominance} introduced by Pilz and Welzl~\cite{DBLP:conf/compgeom/PilzW15,Pilz2017}. Formally, given two sets $P,S$ of $n$ points in the plane, we say that a bijection $f: P \to S$ is \emph{crossing-preserving} if whenever $\overline{pq}$ crosses $\overline{p'q'}$ (where $p,q,p',q'$ are points of $P$) then $\overline{f(p)f(q)}$ crosses $\overline{f(p')f(q')}$ in $S$. If such a mapping exists, we say that $S$ \emph{crossing-dominates} $P$ and we denote it by $P \preceq S$. This relation defines a partial order on \emph{order types} of sets of $n$ points. Pilz and Welzl showed in \cite{Pilz2017} that for various existence and extremal questions it is enough to check the point sets whose order types are maximal in this ordering. In particular, if $S$ crossing-dominates $P$ (with crossing-preserving map $f:P\to S$) and $S$  contains two plane Hamiltonian paths $\pi_1,\pi_2$ that do not include a segment $\overline{ab}$ (where $a,b\in S$), then $f^{-1}(\pi_1),f^{-1}(\pi_2)$ are two plane Hamiltonian paths in $P$ that do not include $\overline{f^{-1}(a)f^{-1}(b)}$. So, to prove that \cref{prop:twopathssegnotinc} holds for all sets of $n$ points, we just need to check that it holds for the maximal ones.  

\twopathssegnotinc*
\label{prop:twopathssegnotinc*}

\begin{proof}
    If $n\ge 7$, the main result in~\cite{KindermannKLV23} guarantees that there exist three edge-disjoint plane Hamiltonian paths on $S$. At most one of these paths contains $\overline{ab}$, so the other two paths satisfy the requirements of the proposition. The proof of the result in~\cite{KindermannKLV23} is constructive and yields a polynomial-time algorithm to construct the three paths.

    For $n=5,6$ we only need to analyze the sets that are maximal with respect to the crossing-dominance order. These four point sets are characterized in \cite{Pilz2017} and can be found in \cref{fig:ordermaximal}, together with paths $\pi_1,\pi_2$ for each (combinatorially different) possible choice of $\overline{ab}$. 
    
    Finally, we observe that, for $n= 4$, two edge-disjoint paths require $6$ segments, so it is impossible for both of the paths to avoid $\overline{ab}$. Similar situations occur for $n\leq 3$.
\end{proof}

\begin{figure}[ht]
    \centering
    \includegraphics{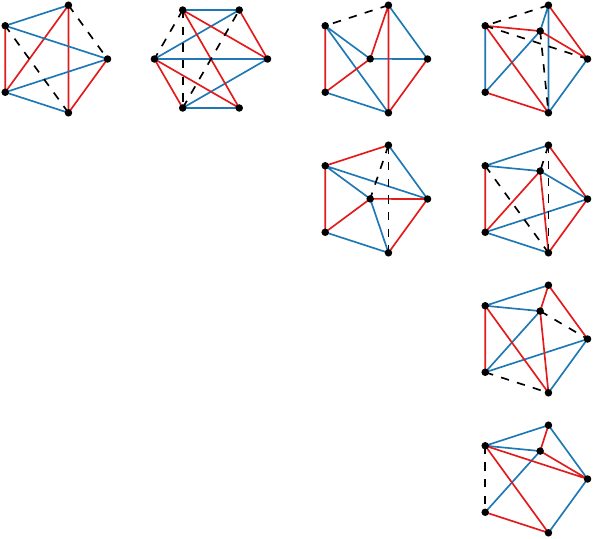}
    \caption{Paths $\pi_1,\pi_2$ satisfying conditions of \cref{prop:twopathssegnotinc} on maximal sets with respect to the crossing-dominance ordering. There are four maximal sets with $5$ or $6$ points (one for each column), and dashed segments indicate possible choices of $\overline{ab}$.}
    \label{fig:ordermaximal}
\end{figure}

\section{Omitted proofs of \cref{sec:onePathWithEdgeOneWithout}}
\label{app:onePathWithEdgeOneWithout}

\pathAlternatingAB*
\label{lem:pathAlternatingAB*}

\begin{proof}
    In~\cite[Theorem 3.1]{AbellanasGHNR99}, it is proved that there exists a plane Hamiltonian path $\pi$ on $S_{\mathrm{down}} \cup S_{\mathrm{up}}$ such that
    every edge in $\pi$ has one endpoint in~$S_{\mathrm{down}}$ and one endpoint in~$S_{\mathrm{up}}$.
    To show that we can find such a path starting in~$z$, and also to analyze the running time,
    let us recall the algorithm used in the proof of that theorem.
    
    Without loss of generality, assume that the separating line is horizontal. Consider the convex hull $\CH(S_{\mathrm{down}} \cup S_{\mathrm{up}})$ of the points $S_{\mathrm{down}} \cup S_{\mathrm{up}}$.
    There are exactly two edges of $\partial \CH(S_{\mathrm{down}} \cup S_{\mathrm{up}})$
    that connect a point from $S_{\mathrm{down}}$ to a point from~$S_{\mathrm{up}}$. Among them, let $ab$, where $a \in S_{\mathrm{down}}$ and $b \in S_{\mathrm{up}}$, be the one intersecting the separating line more to the left. We fix $b$ as the first point of~$\pi$.
    We next repeat this procedure for $\CH(S_{\mathrm{down}} \cup (S_{\mathrm{up}} \setminus \{b\}))$: Let $a'b'$, with $a' \in S_{\mathrm{down}}$ and $b' \in S_{\mathrm{up}}$ be, among the two edges of $\partial \CH(S_{\mathrm{down}} \cup (S_{\mathrm{up}} \setminus \{b\}))$
    connecting a point from $S_{\mathrm{down}}$ to a point from~$S_{\mathrm{up}} \setminus \{b\}$, the one intersecting the separating line more to the left.
    We add $a'$ as the second point of~$\pi$.
    We repeat this process, alternating between points from~$S_{\mathrm{down}}$ and points from~$S_{\mathrm{up}}$,
    until no points remain.
    
    Clearly, $z$ can be any of the two points from~$S_{\mathrm{up}}$
    that belong to the two edges connecting a point from~$S_{\mathrm{down}}$
    to a point from~$S_{\mathrm{up}}$ in $\partial \CH(S_{\mathrm{down}} \cup S_{\mathrm{up}})$.
    
    To implement this algorithm, we use a dynamic convex hull algorithm supporting point deletions.
    Hershberger and Suri describe such a data structure
    that can be constructed in $O(n \log n)$ time
    and has amortized $O(\log n)$ update time~\cite[Theorem 2.6]{HershbergerS92}.
    In each iteration, the edge on the convex hull intersecting the separating line more to the left can be found 
    via a common tangent to $\CH(S_{\mathrm{down}})$ and $\CH(S_{\mathrm{up}})$
    in $O(\log n)$ time~\cite[Theorem 3.2]{HershbergerS92}.
    In total, this yields an $O(n \log n)$-time algorithm.
\end{proof}

\OneEdgeOnePath*
\label{Thm:OneEdgeOnePath*}

\begin{proof}
    It remains to prove the somewhat more technical subcases of Case~B.
    Let us first restate Case~B again and then consider the individual subcases.
    
    \begin{description}
        \item[Case B: $a$ and $b$ are in the same partition.]
        In this case, the main idea is similar as before: to use $\pi$
        as $\pi_2$ (since $\pi$ does not contain $\overline{ab}$),
        to find spanning paths of $S_{\mathrm{up}}$ and $S_{\mathrm{down}}$ including $\overline{ab}$,
        and to connect them as $\pi_1$.
        Let $x$ be a topmost point of~$S_{\mathrm{down}}$, and
        let $y$ be a bottommost point of~$S_{\mathrm{up}}$.
        Clearly, $x$ sees at least two points of~$S_{\mathrm{up}}$, and
        $y$ sees at least two points of~$S_{\mathrm{down}}$.
        The conditions of the following cases can be checked in $O(n \log n)$ time.
        
        \item[Case B1: $x$ sees a point $s$ of $S_{\mathrm{up}}$ such that $\overline{xs} \notin \pi$.]
        We distinguish the following subcases.
        Note that any y-monotone path contains $\overline{ab}$ since $\overline{ab}$ is horizontal.
        In all subcases, an implementation in $O(n \log n)$ time is straight-forward.
        
        \item[Case B1.1: $\{x, s\} \cap \{a, b\} = \emptyset$.]
        We know that by \cref{lem:pathGivenSegment}, there exist two plane Hamiltonian paths $\pi_1^{\mathrm{down}}$ and $\pi_1^{\mathrm{up}}$ on $S_{\mathrm{down}}$ and $S_{\mathrm{up}}$ starting in~$x$ and~$s$, respectively,
        such that one of them includes $\overline{ab}$.
        Let $\pi_1$ be the concatenation of $\pi_1^{\mathrm{down}}$, $\overline{xs}$, and $\pi_1^{\mathrm{up}}$.
        Finally, let $\pi_2 = \pi$.
        Clearly, $\pi_1$ and $\pi_2$ are two plane Hamiltonian paths 
        such that $\overline{ab} \in \pi_1$ and $\pi_1 \cap \pi_2 = \emptyset$.
        
        \item[Case B1.2: $x \in \{a, b\}$.]
        Here, $\overline{ab}$ is the topmost edge of $\partial \CH(S_{\mathrm{down}})$.
        Let $\pi_1^{\mathrm{down}}$ be a path traversing all points
        of~$S_{\mathrm{down}}$ by increasing y-coordinate and ending in~$x$.
        Further, let~$\pi_1^{\mathrm{up}}$ be a path traversing all points
        of~$S_{\mathrm{up}}$ and ending in~$s$;
        such a path exists by \cref{lem:pathTwoEndpoints}.
        Let $\pi_1$ be the concatenation of $\pi_1^{\mathrm{down}}$, $\overline{xs}$, and $\pi_1^{\mathrm{up}}$.
        Finally, let $\pi_2 = \pi$.
        Clearly, $\pi_1$ and $\pi_2$ are two plane Hamiltonian paths 
        such that $\overline{ab} \in \pi_1$ and $\pi_1 \cap \pi_2 = \emptyset$.
        
        \item[Case B1.3: $s \in \{a, b\}$.]
        We know that $y$ sees at least two points of~$S_{\mathrm{down}}$.
        If one of them is in $\pi$ not a neighbor of~$y$, we go to Case~B2.
        Otherwise, $y$ sees $x$ and a vertex~$q$ of $S_{\mathrm{down}}$, and $\overline{xy}, \overline{yq} \in \pi$;
        see \cref{fig:M-a}.
        Let $p$ be the neighbor of~$x$ on~$\pi$ that is distinct from~$y$.
        If $x$ is an endpoint of~$\pi$, then let $p$ be the neighbor of~$q$ on~$\pi$ that is distinct from~$y$;
        we consider only the case where $p$ is a neighbor of~$x$
        because the other case works analogously.
        Let $\pi_2$ be the path obtained from~$\pi$ by concatenating
        the first part of $\pi$ until $p$, $\overline{py}$, $\overline{yx}$, $\overline{xq}$,
        and the last part of $\pi$ starting at~$q$; see \cref{fig:M-b}.
        Observe that this modification does not create intersections within~$\pi_2$
        because the triangles $\triangle xyp$ and $\triangle xyq$ cannot contain a vertex or an edge from~$\pi$ in their interior.
        We let $\pi_1^{\mathrm{down}}$ be a plane path through all points of $S_{\mathrm{down}}$ starting in~$x$
        and ending in~$q$, which exists thanks to \cref{lem:pathTwoEndpoints}.
        It remains to construct the rest of~$\pi_1$.
        See \cref{fig:M-1.3} for an illustration.
        
        \item[Case B1.3a: $p \in \{a, b\}$ or $p$ lies above $\ell(ab)$.]
        Let $\pi_1^{\mathrm{up}}$ be the path through $S_{\mathrm{up}}$ by decreasing y-coordinate.
        Assure that if $p \in \{a, b\}$, then $p$ does not have a neighbor with smaller y-coordinate in~$\pi_1^{\mathrm{up}}$.
        Concatenate $\pi_1^{\mathrm{up}}$ and $\pi_1^{\mathrm{down}}$ via $\overline{yq}$ to get the plane path~$\pi_1$.
        Note that $\pi_1$ and $\pi_2$ do not share edges: $\overline{yq}$ is only in~$\pi_1$,
        $\overline{xq}$ and~$\overline{yp}$ are only in~$\pi_2$.
        
        \item[Case B1.3b: $p$ lies below $\ell(ab)$.]
        Let $\pi_1^{\mathrm{top}}$ be the path through all vertices of $S_{\mathrm{up}}$ above $\ell(ab)$ by decreasing y-coordinate.
        Split the vertices of $S_{\mathrm{up}}$, excluding $a$ and the vertices in~$\pi_1^{\mathrm{top}}$,
        along the line $\ell(xp)$ and name these sets $S_{\mathrm{up}}^{y}$ and $S_{\mathrm{up}}^{p}$
        such that~$S_{\mathrm{up}}^{y}$ contains~$y$ and $S_{\mathrm{up}}^{p}$ contains~$p$.
        Let $\pi_1^{\mathrm{up}, y}$ and $\pi_1^{\mathrm{up}, p}$ be two paths
        over~$S_{\mathrm{up}}^{y}$ and~$S_{\mathrm{up}}^{p}$
        having one endpoint in~$y$ and~$p$, respectively.
        As the other endpoint, choose $b$ if it is in the set, or an arbitrary vertex otherwise.
        Such paths exist due to \cref{lem:pathTwoEndpoints}.
        Assume that $b \in S_{\mathrm{up}}^{y}$; the other case is symmetric.
        Concatenate $\pi_1^{\mathrm{top}}$, $\overline{ab}$, $\pi_1^{\mathrm{up}, y}$, $\overline{yq}$, $\pi_1^{\mathrm{down}}$, $\overline{xp}$, and $\pi_1^{\mathrm{up}, p}$,
        which yields~$\pi_1$.
        Note that $\pi_1$ and $\pi_2$ do not share edges: $\overline{yq}$ and $\overline{xp}$ are only in~$\pi_1$,
        $\overline{xq}$ and~$\overline{yp}$ are only in~$\pi_2$.
        
        \begin{figure}[t]
            \centering
            \begin{subfigure}[t]{.47 \textwidth}
                \centering
                \includegraphics[page=8]{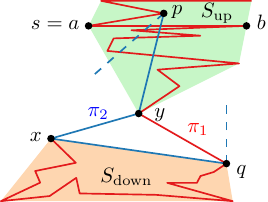}
                \subcaption{}
                \label{fig:M-1.3a}
            \end{subfigure}
            \hfill
            \begin{subfigure}[t]{.47 \textwidth}
                \centering
                \includegraphics[page=7]{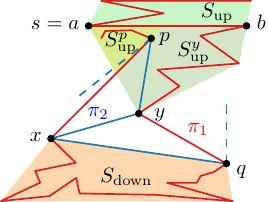}
                \subcaption{}
                \label{fig:M-1.3b}
            \end{subfigure}
            \caption{Situations in Case B1.3a and Case B1.3b.}
            \label{fig:M-1.3}
        \end{figure}
        
        \item[Case B2: $y$ sees a point $s$ of $S_{\mathrm{down}}$ such that $\overline{ys} \notin \pi$.]
        This case is symmetric to Case~B1.

        \item[Case B3: $x$ sees two points $o, p \in S_{\mathrm{up}}$,
        $y$ sees two points $q, r \in S_{\mathrm{down}}$,] \phantom{linebreakafterthis}
        \textbf{and $\overline{ox}, \overline{px}, \overline{qy}, \overline{ry} \in \pi$.}
        ~We start with some observation regarding the arrangement of points
        and line segments in this case.
        First, note that $x$ and $y$ see each other
        as they are topmost and bottommost points of $S_{\mathrm{down}}$ and $S_{\mathrm{up}}$, respectively.
        Hence, we can assume that $o = y$, $r = x$,
        and $\overline{ox} = \overline{xy} = \overline{ry}$.
        We can also assume that $\overline{px}$, $\overline{xy}$, and $\overline{qy}$
        are three distinct line segments that appear in this order along~$\pi$.
        See \cref{fig:M-a}.
        We can assume that, on $\partial \CH(S_{\mathrm{up}})$, $p$ precedes $y$ in counterclockwise order, and
        that, on $\partial \CH(S_{\mathrm{down}})$, $x$ precedes $q$ in clockwise order.
        
        If $p$ and~$q$ were both on the same (mirrored) order (around their convex hulls) relative to~$x$ and~$y$,
        then $\overline{px}$ and $\overline{qy}$ would cross.
        
        \begin{figure}[t]
            \centering
            \begin{subfigure}[t]{.47 \textwidth}
                \centering
                \includegraphics[page=1]{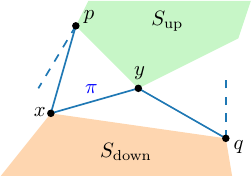}
                \subcaption{$\pi$ uses all segments
                    to points that~$x$ and~$y$ see.}
                \label{fig:M-a}
            \end{subfigure}
            \hfill
            \begin{subfigure}[t]{.47 \textwidth}
                \centering
                \includegraphics[page=2]{Figures/M.pdf}
                \subcaption{We modify~$\pi$ (to~$\pi_2$) so that only
                    $\overline{xy}$ is used.}
                \label{fig:M-b}
            \end{subfigure}
            \caption{Situation in Case B3 (and similarly in Case B1.3):
                $x$ sees only the vertices $p$ and $y$ from~$S_{\mathrm{up}}$ and
                $y$ sees only the vertices $q$ and $x$ from~$S_{\mathrm{down}}$;
                $p$ and $q$ do not necessarily see each other.}
            \label{fig:M}
        \end{figure}
        
        Second, note that $\overline{qx}$ and $\overline{py}$ are line segments
        on $\partial \CH(S_{\mathrm{down}})$ and $\partial \CH(S_{\mathrm{up}})$, respectively,
        that are not a part of~$\pi$
        and that are not crossed by line segments of~$\pi$;
        otherwise, a straight-line segment would leave and enter
        the triangle $\triangle xyp$ or $\triangle xqy$ from the same side.
        By carefully modifying~$\pi$, we can now find a bridge for the path containing $\overline{ab}$.
        Let $\pi_2$ be the path obtained from~$\pi$ by concatenating
        the first part of $\pi$ until $p$, $\overline{py}$, $\overline{yx}$, $\overline{xq}$,
        and the last part of $\pi$ starting at~$q$; see \cref{fig:M-b}.
        It remains to construct~$\pi_1$ but we must be careful in a special case described next.
        
        \item[Case B3.1: $\{x, y\} \cap \{a, b\} \ne \emptyset$.]
        This case implies $\overline{ab} = \overline{qx}$ or $\overline{ab} = \overline{py}$
        as, otherwise, $x$~or~$y$ would see three points of $S_{\mathrm{up}}$ or $S_{\mathrm{down}}$, respectively.
        Assume, without loss of generality, that $\overline{ab} = \overline{qx}$.
        Rename $\pi_2$ to $\pi_1$ since it contains $\overline{ab}$.
        Let $\pi_2^{\mathrm{down}}$ be the path containing the points of $S_{\mathrm{down}} \setminus \{q\}$
        in increasing order of y-coordinate and
        let $\pi_2^{\mathrm{up}}$ be a plane path through all points of $S_{\mathrm{up}}$ starting in~$p$
        and ending in~$y$, which exists thanks to \cref{lem:pathTwoEndpoints}.
        Let $\pi_2$ be the concatenation of $\pi_2^{\mathrm{down}}$, $\overline{px}$, $\pi_2^{\mathrm{up}}$, and~$\overline{qy}$;
        see \cref{fig:M-c}.
        Clearly, $\pi_1$ and $\pi_2$ are two plane Hamiltonian paths 
        such that $\overline{ab} \in \pi_1$ and $\pi_1 \cap \pi_2 = \emptyset$, and the construction can be done in $O(n \log n)$ time.
        
        \begin{figure}[t]
            \centering
            \begin{subfigure}[t]{.47 \textwidth}
                \centering
                \includegraphics[page=5]{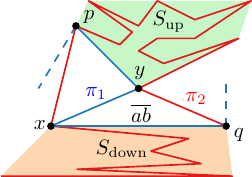}
                \subcaption{Case B3.1.}
                \label{fig:M-c}
            \end{subfigure}
            \hfill
            \begin{subfigure}[t]{.47 \textwidth}
                \centering
                \includegraphics[page=6]{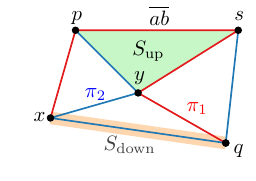}
                \subcaption{Case B3.2.}
                \label{fig:M-d}
            \end{subfigure}
            
            \bigskip
            
            \begin{subfigure}[t]{.47 \textwidth}
                \centering
                \includegraphics[page=3]{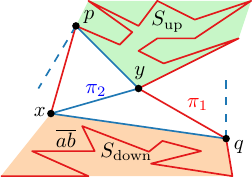}
                \subcaption{Case B3.3a.}
                \label{fig:M-e}
            \end{subfigure}
            \hfill
            \begin{subfigure}[t]{.47 \textwidth}
                \centering
                \includegraphics[page=4]{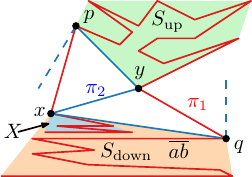}
                \subcaption{Case B3.3b.}
                \label{fig:M-f}
            \end{subfigure}
            
            \caption{Adding the second path according to the subcases of Case~B3.}
            \label{fig:M-2}
        \end{figure}
        
        \item[Case B3.2: $\{x, y\} \cap \{a, b\} = \emptyset$ and $|S_{\mathrm{down}}| = 2$.]
        Let $\pi_1 = \langle x, p, s, y, q \rangle$
        where $s$ is the remaining point of $S_{\mathrm{up}}$; see \cref{fig:M-d}.
        Clearly, $\overline{ab} = \overline{ps}$; thus, $\pi_1$ contains $\overline{ab}$.
        Further, $\pi_1$ and~$\pi_2$ are two plane Hamiltonian paths
        and $\pi_1 \cap \pi_2 = \emptyset$.
        $S$ has only constant size.
        
        \item[Case B3.3: $\{x, y\} \cap \{a, b\} = \emptyset$ and $|S_{\mathrm{down}}| \ge 3$.]
        We assume that $a, b \in S_{\mathrm{down}}$; the case $a, b \in S_{\mathrm{up}}$ is symmetric.
        We further distinguish whether $q \in \{a, b\}$ or not.
        
        \item[Case B3.3a: $q \notin \{a, b\}$.]
        Let $\pi_1^{\mathrm{down}}$ be a plane Hamiltonian path of $S_{\mathrm{down}} \setminus \{x\}$
        such that it includes $\overline{ab}$ and has one endpoint in~$q$,
        which exists according to \cref{lem:pathGivenSegment}
        and can be found in $O(n \log n)$ time.
        Let $\pi_1^{\mathrm{up}}$ be a plane path through all points of $S_{\mathrm{up}}$ starting in~$y$
        and ending in~$p$, which exists according to \cref{lem:pathTwoEndpoints}
        and can be found in $O(n \log n)$ time.
        Let $\pi_1$ be the concatenation of $\pi_1^{\mathrm{down}}$, $\overline{qy}$, $\pi_1^{\mathrm{up}}$, and~$\overline{px}$;
        see \cref{fig:M-e}.
        Clearly, $\pi_1$ and $\pi_2$ are two plane Hamiltonian paths 
        such that $\overline{ab} \in \pi_1$ and $\pi_1 \cap \pi_2 = \emptyset$.
        
        \item[Case B3.3b: $q \in \{a, b\}$.]
        We define $X$ as the subset of points in~$S_{\mathrm{down}}$
        whose y-coordinate is strictly greater than the y-coordinate of~$q$;
        in particular $x \in X$.
        Let $\pi_1^{\mathrm{down}}$ be the path containing all points of $S_{\mathrm{down}} \setminus X$
        in increasing order of y-coordinate and ending in~$q$.
        Note that $\overline{ab} \in \pi_1^{\mathrm{down}}$.
        Let $\pi_1^{\mathrm{up}}$ a plane path through all points of $S_{\mathrm{up}}$ starting in~$y$
        and ending in~$p$, which exists according to \cref{lem:pathTwoEndpoints}.
        Let $\pi_1^X$ be the path starting in~$x$ and
        containing all points of~$X$ in decreasing order of y-coordinate.
        Let $\pi_1$ be the concatenation of $\pi_1^{\mathrm{down}}$, $\overline{qy}$, $\pi_1^{\mathrm{up}}$, $\overline{px}$, and~$\pi_1^X$;
        see \cref{fig:M-f}.
        Clearly, $\pi_1$ and $\pi_2$ are two plane Hamiltonian paths 
        such that $\overline{ab} \in \pi_1$ and $\pi_1 \cap \pi_2 = \emptyset$
        and the construction can be done in $O(n \log n)$ time.
    \end{description}
    
    With this case distinction, we have shown that two such paths can in any situation be found.
    This completes the proof.
\end{proof}

\section{Omitted proofs of \cref{sec:twoPathsWithEdge}}
\label{app:twoPathsWithEdge}

\onesidedcase*
\label{lem:one-sidedcase*}

\begin{proof}
    We assume that the segment $\overline{ab}$ is horizontal. We express $S= S_1 \cup S_2 \cup S_3 \cup \{a,b\}$ where $S_1,S_2,S_3, \{a,b\}$ are pairwise disjoint as follows. We define $S_3$ to be the inclusion-maximal set of points in $S\setminus \{a,b\}$ such that at least one segment between points of $S_3$ is a bridge over $\ell(ab)$, and no segment between two points in $S_3$ crosses $\overline{ab}$. Then we define $S_1$ to be the set of points in $S\setminus (S_3 \cup \{a,b\})$ which lie above the line $\ell(ab)$ and $S_2$ the set of points in $S\setminus (S_3 \cup \{a,b\})$ which lie below the line $\ell(ab)$. We remark that the definition of $S_3$ in this decomposition is not unique. 
    Further, we write $S_3'$ and $S_3''$ for the subsets of $S_3$ lying above and below $\ell(ab)$, respectively, see \cref{fig:singlesidedfig-1}.
    
    Note that if $|S_1 \cup S_3'|$ and $|S_2 \cup S_3''|$ satisfy the conditions from \cref{lem:diagonalcase}, we can always find the two paths. We just treat $\overline{ab}$ as a diagonal and we don't use any edges that cross $\ell(ab)$. Therefore, we only need to consider the cases where $|S_1 \cup S_3'|\in \{2,3\}$ or $|S_2 \cup S_3''|\in \{2,3\}$. 
    
    We start by proving that in certain cases it is impossible to find $\pi_1$ and $\pi_2$.
    
    \begin{description}
    \item[Case A: $|S_1|=|S_2|=|S_3'|=|S_3''|=1.$]
    Let $s_1,s_2,s_3'$ and $s_3''$ be the points in $S_1,S_2,S_3'$ and $S_3''$, respectively. Assume that we can find two paths $\pi_1, \pi_2$ with the desired properties. Then we can assume that $\pi_1$ contains the segment $\overline{s_3's_3''}$, since otherwise we could also find two such paths in the case where $a$ and $b$ are on the convex hull of $S$. Further, we can assume that $\pi_1 $ contains the segment $\overline{s_3's_1}$. However, this implies that $\pi_2$ needs to contain the segments $\overline{as_1}$ and $\overline{bs_3'}$, which means that in order for $\pi_2$ to cover all of the points of $S$, we need to either cross $\overline{ab}$ or one of $a,b$ needs to have degree $3$, which is impossible.
    
    \item[Case B: $|S_1|=|S_3'|=|S_3''|=1, |S_2| = 2$.]
    Assume that we can find two paths $\pi_1, \pi_2$ with the desired properties. Denote as previously the vertices of the bridge by $s_3'$ and $s_3''$ and let $S_1=\{s_1\}$ and $S_2 = \{s_2',s_2''\}$. As in the previous case, we argue that one of the paths has to contain the bridge, say $\pi_1$. Then in $\pi_2$ we cannot cross $\ell(ab)$ and hence $\pi_2$ contains the segment $\overline{s_3's_1}$ and two of the segments $\overline{s_3''s_2'}, \overline{s_3''s_2''},\overline{s_2's_2''}$. Therefore we can assume, without loss of generality, that $\pi_2 = s_3's_1,a,b,s_2',s_2'',s_3''$. Now we consider $\pi_1$; it needs to contain the segment $\overline{s_1b}$ and exactly one of the segments $\overline{as_2'}$ or $\overline{as_3''}$. Lastly, it has to contain the segment $\overline{s_2's_3''}$. This gives a total of $5$ segments and no other can be used - which is impossible, since we assumed that $\pi_1$ is a Hamiltonian path on $7$ points.

    \item[Case C: $|S_3'|=|S_3''|=1, |S_1|=|S_2| = 2$.]
    Assume that we can find two paths $\pi_1, \pi_2$ with the desired properties. Denote the vertices as before (with the distinction of $S_1$ now comprising of $s_1'$ and $s_1''$). As in the previous case, we argue that one of the paths has to contain the bridge, say $\pi_1$. As in the previous case, the segment $\overline{ab}$ separates $S_1$ and $S_2$ and therefore $\pi_2$ needs to have endpoints in different halfplanes determined by $\ell(ab)$. Therefore, the order of the vertices in $\pi_2$ is as in \cref{fig:bluepath}. It follows that $\pi_1$ can contain at most one edge from each of $\triangle s_3's_1's_1''$ and $\triangle s_3''s_2's_2''$. Points $a$ and $b$ can have at most one other edge incident to them, which in total gives us a maximum of $6$ edges for~$\pi_1$, which is a contradiction since we assumed it is a Hamiltonian path on $8$ points.
    \end{description}
    
    We claim that in all of the other cases, we can find $\pi_1$ and $\pi_2$ with the desired properties. This follows from a long analysis where we find paths manually for a lot of small cases. For the summary of the result, see \cref{tab: singlesidedtable} and \cref{fig:singlesidedfig}.
    
    \begin{figure}[t]
        \centering
        \includegraphics{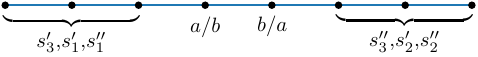}
        \caption{The order of points in the path $\pi_2$ in Case 3 of the proof of \cref{lem:one-sidedcase}.}
        \label{fig:bluepath}
    \end{figure}

    \begin{figure}[p]
    \centering
    \begin{subfigure}[b]{0.3\textwidth}
        \centering
        \includegraphics[page=17]{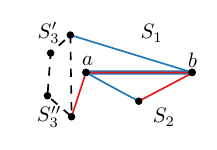}
        \subcaption{$S_1=\emptyset, |S_2|=1, |S_3' \cup S_3''|\ge 4$.}
        \label{fig:exceptions-1}
    \end{subfigure}
    \hfill
    \begin{subfigure}[b]{0.3\textwidth}
        \centering
        \includegraphics[page=18]{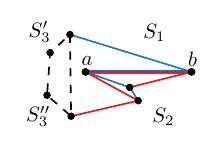}
        \subcaption{$S_1=\emptyset, |S_2|=2, |S_3' \cup S_3''|\ge 4$.}
        \label{fig:exceptions-2}
    \end{subfigure}
    \hfill
    \begin{subfigure}[b]{0.3\textwidth}
        \centering
        \includegraphics[page=19]{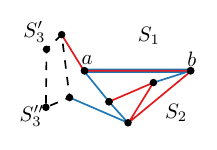}
        \subcaption{ $S_1=\emptyset, |S_2|=3, |S_3' \cup S_3''|\ge 4$.}
        \label{fig:exceptions-3}
    \end{subfigure}
    
    \vspace{1em}
    
    \begin{subfigure}[b]{0.3\textwidth}
        \centering
        \includegraphics[page=20]{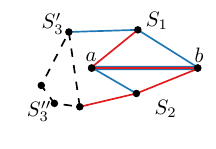}
        \subcaption{$|S_1|=1, |S_2|=1, |S_3' \cup S_3''|\ge 4$.}
        \label{fig:exceptions-4}
    \end{subfigure}
    \hfill
    \begin{subfigure}[b]{0.3\textwidth}
        \centering
        \includegraphics[page=21]{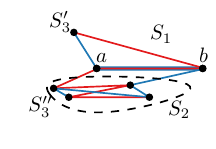}
        \subcaption{$S_1=\emptyset, |S_2|=2, |S_3'|=1,  |S_3''|\ge 2$.}
        \label{fig:exceptions-5}
    \end{subfigure}
    \hfill
    \begin{subfigure}[b]{0.3\textwidth}
        \centering
        \includegraphics[page=22]{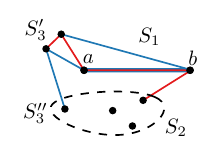}
        \subcaption{$S_1=\emptyset, |S_2|\ge 3, |S_3'|=2,  |S_3''|=1$.}
        \label{fig:exceptions-6}
    \end{subfigure}
    
    \caption{Special cases in the proof of \cref{lem:one-sidedcase} which are not included in \cref{tab: singlesidedtable}. We apply \cref{Lem: KKGVThm2} or \cref{Lem: S=4} in the dashed regions.}
    \label{fig:exceptions}
    \end{figure}

\begin{table}[p]
    \centering
    \begin{tabular}{|l|l|l|l|l|l|}
        \hline
        \rowcolor[HTML]{FFFFC7} 
        $|S_1|$ & $|S_2|$ & $|S_3'|$ & $|S_3''|$ & $\pi_1,\pi_2$ exist & Reference                                                                  \\ \hline
        $0$     & $0$     & $\ge 1$  & $\ge 1$   & YES                 & Reduce to the setting of \cref{lem:diagonalcase}  \\ \hline
        $0$     & $1$     & $1$      & $1$       & YES                 & \cref{fig:singlesidedfig-2}              \\ \hline
        $0$     & $1$     & $1$      & $2$       & YES                 & \cref{fig:singlesidedfig-3}              \\ \hline
        $0$     & $1$     & $2$      & $1$       & YES                 & \cref{fig:singlesidedfig-4}              \\ \hline
        $0$     & $2$     & $1$      & $1$       & YES                 & \cref{fig:singlesidedfig-5}              \\ \hline
        $0$     & $2$     & $2$      & $1$       & YES                 & \cref{fig:singlesidedfig-6}              \\ \hline
        $1$     & $1$     & $1$      & $1$       & NO                  & Case $1$ in the proof of \cref{lem:one-sidedcase} \\ \hline
        $1$     & $1$     & $2$      & $1$       & YES                 & \cref{fig:singlesidedfig-7}              \\ \hline
        $1$     & $2$     & $1$      & $1$       & NO                  & Case $2$ in the proof of \cref{lem:one-sidedcase} \\ \hline
        $1$     & $2$     & $1$      & $2$       & YES                 & \cref{fig:singlesidedfig-8}              \\ \hline
        $1$     & $2$     & $2$      & $1$       & YES                 & \cref{fig:singlesidedfig-9}              \\ \hline
        $1$     & $2$     & $3$      & $1$       & YES                 & \cref{fig:singlesidedfig-10}              \\ \hline
        $1$     & $2$     & $1$      & $3$       & YES                 & \cref{fig:singlesidedfig-11}              \\ \hline
        $1$     & $2$     & $\ge 2$  & $\ge 3$   & YES                 & \cref{fig:singlesidedfig-12}             \\ \hline
        $2$     & $2$     & $1$      & $1$       & NO                  & Case $3$ in the proof of \cref{lem:one-sidedcase} \\ \hline
        $2$     & $2$     & $\ge 2$  & $\ge 1$    & YES                 & \cref{fig:singlesidedfig-13}              \\ \hline
        $3$     & $2$     & $\ge 1$   & $\ge 1$    & YES                 & \cref{fig:singlesidedfig-14}              \\ \hline

    \end{tabular}
    \bigskip
    
    \caption{All of the cases from the proof of \cref{lem:one-sidedcase}.}
    \label{tab: singlesidedtable}
\end{table}

\begin{figure}[p]
    \centering
    \begin{subfigure}[b]{0.3\textwidth}
        \centering
        \includegraphics[page=2]{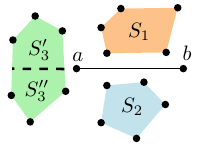}
        \subcaption{ General setup in the proof of \cref{lem:one-sidedcase}.}
        \label{fig:singlesidedfig-1}
    \end{subfigure}
    \hfill
    \begin{subfigure}[b]{0.3\textwidth}
        \centering
        \includegraphics[page=3]{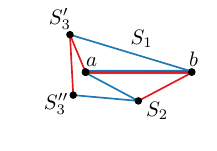}
        \subcaption{Case "0111".}
        \label{fig:singlesidedfig-2}
    \end{subfigure}
    \hfill
    \begin{subfigure}[b]{0.3\textwidth}
        \centering
        \includegraphics[page=4]{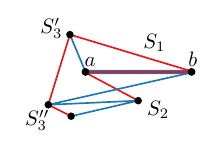}
        \subcaption{Case "0112".}
        \label{fig:singlesidedfig-3}
    \end{subfigure}

    \begin{subfigure}[b]{0.3\textwidth}
        \centering
        \includegraphics[page=5]{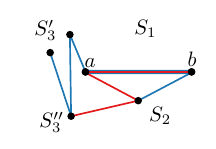}
        \subcaption{Case "0121".}
        \label{fig:singlesidedfig-4}
    \end{subfigure}
    \hfill
    \begin{subfigure}[b]{0.3\textwidth}
        \centering
        \includegraphics[page=6]{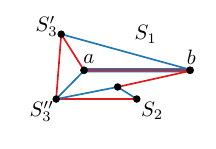}
        \subcaption{Case "0211".}
        \label{fig:singlesidedfig-5}
    \end{subfigure}
    \hfill
    \begin{subfigure}[b]{0.3\textwidth}
        \centering
        \includegraphics[page=7]{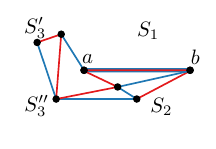}
        \subcaption{Case "0221".}
        \label{fig:singlesidedfig-6}
    \end{subfigure}
    
    \begin{subfigure}[b]{0.3\textwidth}
        \centering
        \includegraphics[page=8]{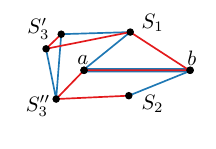}
        \subcaption{ Case "1121".}
        \label{fig:singlesidedfig-7}
    \end{subfigure}
    \hfill
    \begin{subfigure}[b]{0.3\textwidth}
        \centering
        \includegraphics[page=9]{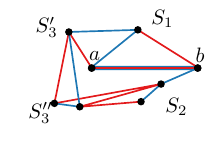}
        \subcaption{Case "1212".}
        \label{fig:singlesidedfig-8}
    \end{subfigure}
    \hfill
    \begin{subfigure}[b]{0.3\textwidth}
        \centering
        \includegraphics[page=10]{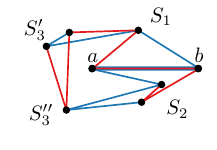}
        \subcaption{Case "1221".}
        \label{fig:singlesidedfig-9}
    \end{subfigure}
    
    \begin{subfigure}[b]{0.3\textwidth}
        \centering
        \includegraphics[page=11]{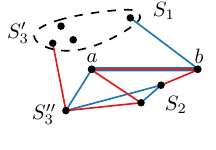}
        \subcaption{ Case "1231".}
        \label{fig:singlesidedfig-10}
    \end{subfigure}
    \hfill
    \begin{subfigure}[b]{0.3\textwidth}
        \centering
        \includegraphics[page=12]{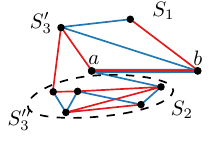}
        \subcaption{Case "1213".}
        \label{fig:singlesidedfig-11}
    \end{subfigure}
    \hfill
    \begin{subfigure}[b]{0.3\textwidth}
        \centering
        \includegraphics[page=13]{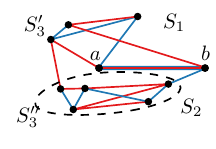}
        \subcaption{Case "1223".}
        \label{fig:singlesidedfig-12}
    \end{subfigure}
    
    \begin{subfigure}[b]{0.3\textwidth}
        \centering
        \includegraphics[page=14]{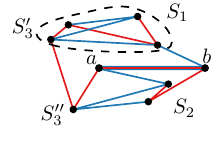}
        \subcaption{Case "2221".}
        \label{fig:singlesidedfig-13}
    \end{subfigure}
    \hfill
    \begin{subfigure}[b]{0.3\textwidth}
        \centering
        \includegraphics[page=15]{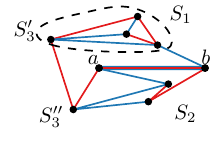}
        \subcaption{Case "3211".}
        \label{fig:singlesidedfig-14}
    \end{subfigure}
    
    \caption{Construction of $\pi_1$ and $\pi_2$ when $a$ is not a vertex of $\CH(S)$ and  $b$ is.}
    \label{fig:singlesidedfig}
\end{figure}

We remark that to solve some of the cases we employ \cref{Lem: KKGVThm2,Lem: S=4}. We do this by finding a subset $P$ of $S\setminus \{a,b\}$ of size at least $4$, and then finding two paths $\pi_1',\pi_2'$ in $S \setminus P$ which satisfy the conditions of the lemma and whose endpoints see two distinct vertices $x,y$ of the convex hull of $P$. We then apply either \cref{Lem: KKGVThm2} or \cref{Lem: S=4} to $P$ and find two more paths $\pi_1^x, \pi_2^y$. Finally, we connect $\pi_1'$ to $\pi_1^x$ and $\pi_2'$ to $y$ to obtain our two paths $\pi_1$ and $\pi_2$. Having that in mind, we note that in the following cases paths exist (and are shown in corresponding figures) and we exclude them from \cref{tab: singlesidedtable} to keep it more compact: 

\begin{itemize}
    \item $|S_1| = 0, |S_2| = 1, |S_3| \ge 4$ -  \cref{fig:exceptions-1};
    \item $|S_1| = 0, |S_2| = 2, |S_3| \ge 4$ - \cref{fig:exceptions-2};
    \item $|S_1| = 0, |S_2 \cup S_3''| \ge 4,$ - \cref{fig:exceptions-5} and \cref{fig:exceptions-6}.
    \item $|S_1| = 1, |S_2| = 1, |S_3| \ge 4,$ - \cref{fig:exceptions-4};
    \item $|S_1| = 0, |S_2| = 3, |S_3| \ge 4$ - \cref{fig:exceptions-3}. \qed
\end{itemize}

\end{proof}

\twosidedcase*
\label{lem:two-sidedcase*}

\begin{proof}
      We assume that the segment $\overline{ab}$ is horizontal and that $a$ lies to the left of $b$. We define pairwise disjoint sets $S_1,S_2,S_3,S_4$ such that $S= S_1 \cup S_2 \cup S_3 \cup S_4 \cup \{a,b\}$ as follows:
    We define $S_3$ to be the inclusion-maximal set of points in $S \setminus \{a,b\}$ such that at least one segment between points of $S_3$ is a bridge over $\ell(ab)$, no segment between points of $S_3$ crosses the segment $\overline{ab}$, and all of the bridges defined by points of $S_3$ cross $\ell(ab)$ to the left of $a$.
    We define $S_4 \subset S \setminus (\{a,b\} \cup S_3)$ analogously with all bridges defined by points of $S_4$ crossing $\ell(ab)$ to the right of~$b$.
   Further, we write $S_3'$ and $S_3''$ ($S_4'$ and $S_4''$)  for the subsets of $S_3$ ($S_4$) lying above and below $\ell(ab)$ respectively .
    Further, we define $S_1,S_2$ as in the proof of \cref{lem:one-sidedcase}. See \cref{fig:twosidedcasefig-1} for an illustration of this decomposition.
    Note that this partition is not unique and it may not even capture the full structure of the point set we are working with, since we might ignore many edges that we can use to construct our desired paths $\pi_1,\pi_2$. However, we can prove that, even in this restricted setting, we  can always find the two paths. 
    Similarly as before, we only need to investigate the cases in which we are not able to find $\pi_1,\pi_2$ in the setting of \cref{lem:one-sidedcase}. 
    
    \begin{description}
        \item[Case A: $|S_1\cup S_3' \cup S_4'|=|S_2\cup S_3'' \cup S_4''| = 2$.]
        There are few subcases to consider, but the only one where both $S_3$ and $S_4$ are nonempty is the case when $|S_3'|=|S_3''|=|S_4'|=|S_4''|=1,$ seen in \cref{fig:twosidedcasefig-2}. If either $S_3$ or $S_4$ are empty, then the existence or non-existence of the two paths follows by  \cref{lem:diagonalcase,lem:one-sidedcase}.
        
        \item[Case B: $|S_1\cup S_3' \cup S_4'|=3$ and $|S_2\cup S_3'' \cup S_4''| = 2$.]
        There are a few subcases to consider, but as before we only care about the ones where both $S_3$ and $S_4$ are nonempty. Further, we will only consider the case with minimal number of bridges to the left of $a$ and to the right of $b$, which can be seen in \cref{fig:twosidedcasefig-3}. All the cases with more bridges than these have can be solved in the same way.
        
        \item[Case C: $|S_1\cup S_3' \cup S_4'|=3$ and $|S_2\cup S_3'' \cup S_4''| = 3.$]
        Again, we will only consider the case with minimal number of bridges to the left of $a$ and to the right of $b$, which can be seen in \cref{fig:twosidedcasefig-4}. \qed
    \end{description}
\end{proof}

\begin{figure}[t]
    \centering
    \begin{subfigure}[t]{0.45\textwidth}
        \centering
        \includegraphics[page=1]{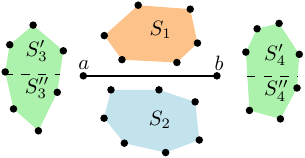}
        \subcaption{General setup in the proof of \cref{lem:two-sidedcase}.}
        \label{fig:twosidedcasefig-1}
    \end{subfigure}
    \hfill
    \begin{subfigure}[t]{0.5\textwidth}
        \centering
        \includegraphics[page=2]{Figures/two-sided-all.pdf}
        \subcaption{Case A.}
        \label{fig:twosidedcasefig-2}
    \end{subfigure}
    
    \bigskip
    
    \begin{subfigure}[t]{0.45\textwidth}
        \centering
        \includegraphics[page = 3]{Figures/two-sided-all.pdf}
        \subcaption{Case B. The red path may change if $S_1$ lies above the line spanned by the points in $S'_3$ and $S''_3$.}
        \label{fig:twosidedcasefig-3}
    \end{subfigure}
    \hfill
    \begin{subfigure}[t]{0.5\textwidth}
        \centering
        \includegraphics[page=4]{Figures/two-sided-all.pdf}
        \subcaption{Case C. The red path may change if $S_1$ is above the line spanned by $S'_3$ and $S'_4$ or if $S_2$ is below the line spanned by $S''_3$ and $S''_4$.}
        \label{fig:twosidedcasefig-4}
    \end{subfigure}
    
    \caption{Construction of $\pi_1$ and $\pi_2$ when neither $a$ nor $b$ are vertices of $\partial \CH(S)$.}
    \label{fig:twosidedcasefig}
\end{figure}

\end{document}